\documentclass[conference]{IEEEtran}\IEEEoverridecommandlockouts
\usepackage{etoolbox}
\makeatletter
\patchcmd{\@makecaption}
  {\scshape}
  {}
  {}
  {}
\makeatletter
\patchcmd{\@makecaption}
  {\\}
  {.\ }
  {}
  {}
\makeatother
 \usepackage{amsmath,amssymb}
 \usepackage{subfigure}
 \usepackage{graphicx,graphics,color,psfrag}
 \usepackage{cite,balance}
 \usepackage{caption}
 \allowdisplaybreaks
 \usepackage{algorithm}
 \usepackage{accents}
 \usepackage{amsthm}
 \usepackage{bm}
 \usepackage{algorithmic}
 \usepackage[english]{babel}
 \usepackage{multirow}
 \usepackage{enumerate}
 \usepackage{cases}
 \usepackage{stfloats}
 \usepackage{dsfont}
 \usepackage{color,soul}
 \usepackage{amsfonts}
 \usepackage{cite,graphicx,amsmath,amssymb}
 \usepackage{subfigure}
 \usepackage{fancyhdr}
 \usepackage{hhline}
 \usepackage{graphicx,graphics}
 \usepackage{array,color}
 \usepackage{amsmath}
\usepackage{float}
\usepackage{amssymb}
\usepackage{amsmath}
\usepackage{amsthm}
\usepackage{amsfonts}
\usepackage{graphicx}
\usepackage{algorithm}
\usepackage{algorithmic}
\usepackage{epstopdf}
\usepackage{cite}
\usepackage{amsmath,bm}
\usepackage{subfigure}
\usepackage{graphicx}
\usepackage{color}
\usepackage{graphicx}
\usepackage{calc}
\usepackage{caption}
\newtheorem{theorem}{\textbf{Theorem}}

\newtheorem{lemma}{\textbf{Lemma}}

\newtheorem{remark}{\textbf{Remark}}

\newtheorem{Prob}{\textbf{Problem}}

\newtheorem{property}{Property}
\begin{document}
%\bibliographystyle{IEEEtran}
%\addtolength{\textfloatsep}{-20pt}
%\title{Minimizing bandwidth requirements for Mobile Edge Computing and Caching Systems}
\title{Bandwidth Gain from Mobile Edge Computing and Caching in Wireless Multicast Systems}
\author{Yaping Sun, Zhiyong Chen, Meixia Tao, \emph{IEEE Fellow} and Hui Liu, \emph{IEEE Fellow}\thanks{The paper was presented in part at IEEE Globecom 2018 \cite{gc18}. 

Y. Sun and M. Tao are with the Department of Electronic Engineering, Shanghai Jiao Tong University, Shanghai 200240, China (e-mail: yapingsun@sjtu.edu.cn; mxtao@sjtu.edu.cn).

Z. Chen (\emph{corresponding author}) and H. Liu are with the Cooperative Medianent Innovation Center, Shanghai Jiao Tong University, Shanghai 200240, China, and also with the Shanghai Key Laboratory of Digital Media Processing and Transmissions, Shanghai 200240, China (e-mail: zhiyongchen@sjtu.edu.cn; huiliu@sjtu.edu.cn).}}
%\title{Modeling and Trade-off for Mobile Communication, Computing and Caching Networks}
%\author{Yaping Sun, Zhiyong Chen, Meixia Tao and Hui Liu\\
%Department of Electronic Engineering, Shanghai Jiao Tong University,  Shanghai,  P. R. China\\
%Email: \{yapingsun, zhiyongchen, mxtao, huiliu\}@sjtu.edu.cn}
%\THANKS{tHIS WORK IS SUPPORTED BY THE nATIONAL nATURAL sCIENCE fOUNDATION OF cHINA UNDER GRANTS 61671291, 61571299 AND 61521062.}}
%This paper was partially supported by the National Natural Science Foundation of China (Grant No. 61671291, 61528103, and 61521062).}
\maketitle
\begin{abstract}
In this paper, we present a novel mobile edge computing (MEC) model where the MEC server has the input and output data of all computation tasks and communicates with multiple caching-and-computing-enabled mobile devices via a shared wireless link. %Each mobile device can pre-store the input or output data of a task and also execute a task locally. 
Each task request can be served from local output caching, local computing with input caching, local computing or MEC computing, each of which incurs a unique bandwidth requirement of the multicast link. Aiming to minimize the transmission bandwidth, we design and optimize the local caching and computing policy at mobile devices subject to latency, caching, energy and multicast transmission constraints. %We aim to investigate the impact of local caching and computing at mobile devices as well as content-centric multicast transmission on the saving of required bandwidth on the wireless link. To this end, we first formulate a joint caching and computing decision optimization problem to minimize the required transmission bandwidth subject to latency, caching and energy constraints at each mobile device in the general case. 
The joint policy optimization problem is shown to be NP-hard. When the output data size is smaller than the input data size, we reformulate the problem as minimization of a monotone submodular function over matroid constraints and obtain the optimal solution via a strongly polynomial  algorithm of Schrijver. On the other hand, when the output data size is larger than the input data size, by leveraging sample approximation and concave convex  procedure together with the alternating direction method of multipliers, we propose a low-complexity high-performance algorithm and prove it converges to a stationary point. Furthermore, we \textcolor{black}{theoretically reveal} how much bandwidth gain can be achieved from computing and caching resources at mobile devices or the multicast transmission for symmetric case. Our results indicate that exploiting the computing and caching resources at mobile devices \textcolor{black}{as well as} multicast transmission can provide significant bandwidth savings.

%\textcolor{red}{Computation task service delivery in a computing-enabled and caching-aided multi-user mobile edge computing (MEC) system is studied in this paper, where a MEC server can deliver the input or output datas of tasks to mobile devices over a wireless multicast channel. The  mobile devices are enabled to store the input or output datas of some tasks, and also compute some tasks locally, reducing the wireless bandwidth consumption. The corresponding framework of this system is established, and under the latency constraint, we jointly optimize the caching and computing policy at mobile devices to minimize the required transmission bandwidth. The joint policy optimization problem is shown to be NP-hard, and based on equivalent transformation and exact penalization of the problem, a stationary point is obtained via concave convex procedure (CCCP).
%Moreover, in a symmetric scenario, gains offered by this approach are derived to analytically understand the influences of caching and computing resources at mobile devices, multicast transmission, the number of mobile devices, as well as the number of tasks on the transmission bandwidth. Our results indicate that exploiting the computing and caching resources at mobile devices can provide significant bandwidth savings.}
\end{abstract}
\section{Introduction}\label{I}
\subsection{\textcolor{black}{Motivation}}
\textcolor{black}{The accelerated acquisition of smart mobile devices brings the proliferation of new kinds of mobile traffic, such as virtual reality (VR) and augmented reality (AR) traffic. According to Cisco's prediction, VR traffic and AR traffic will increase by 11 and 7 times in the next five years, respectively \cite{Cisco}. These modern traffic loads require delivery of huge data and intensive computation at ultra low latency, thereby imposing significant stress on wireless network and incurring severe spectrum scarcity problem. For example, mobile VR delivery requires transmission rate on the order of G bit/s \cite{E}. Therefore, how to tackle the spectrum scarcity problem and save bandwidth becomes one of the most important issues of the network operators.}
%Bandwidth saving is an eternal topic in wireless communications systems, especially in the era of shortage of wireless spectrum resource. 

Recently, \emph{moible edge computing (MEC), caching} and \emph{multicast} have been envisioned as three efficient \textcolor{black}{and promising} approaches to tackle the wireless spectrum crunch problem \cite{3C,multicast}. \textcolor{black}{Specifically, modern data traffic generally incurs heavy computation tasks, e.g., mobile VR/AR delivery and online gaming  \cite{MEC_magazine}. Mobile edge computing offers a promising paradigm to improve user-perceived quality of experience (QoE)  by computing some post-processing low-complexity tasks closer to users, either at the mobile edge server or at the mobile device. In the meantime, modern data traffic exhibits a high degree of asynchronous content reuse \cite{Caire}. Content caching is an effective way to reduce peak time traffic by prefetching contents closer to uses during off peak time, thereby alleviating the backhaul capacity requirement and improving user-perceived QoE in wireless networks. In addition, multicast transmission provides an efficient capacity-offloading approach for common content delivery to multiple subscribers on a same resource block \cite{multicast}.}

%Specifically, modern data traffic exhibits a high degree of asynchronous content reuse \cite{Caire}. Caching popular contents closer to users, e.g., at \textcolor{black}{base stations (BSs)} or even at end users, \textcolor{black}{can reduce traffic load and hence} is gradually recognized as a promising approach to improve bandwidth utilization \cite{caching_magazine}. In addition, modern data traffic also incurs heavy computation tasks \cite{MEC_magazine}. By exploiting the difference of data size before and after calculation, computing the tasks closer to users, e.g., at the BSs or even at end users, can also reduce the required transmission bandwidth \cite{Sunvr}. Furthermore, multicast transmission can deliver popular contents to multiple users on a common bandwidth resource simultaneously, saving the bandwidth over unicast transmission \cite{multicast}. On the other hand, 

\textcolor{black}{Meanwhile, today's smart mobile devices possess a tremendous amount of computing power and storage space. Taking the latest generation of iPhone (iPhone XS) for example, iPhone XS not only has a six-core CPU and a quad-core GPU, but also has the neural engine with 8-core with machine learning core processor. Along with powerful computing capability, iPhone XS also offers 512 GB inbuilt storage, similar to the storage size of a computer in 2010. In light of the above mentioned benefits from mobile edge computing, caching and multicast, we would like to take full advantage of the computing power and caching space available in the mobile devices as well as multicast transmission in a multi-user MEC system to save bandwidth cost.}

%\textcolor{blue}{How to efficiently utilize the computing and caching resources at the mobile devices requires careful design.}
The main challenge of utilizing the mobile device's computing and caching resources is how to design the computing and caching policy for the mobile devices. One particular example is mobile VR delivery \cite{Sunvr}. In the VR framework, the projection component can be computed at the MEC server or at the mobile VR devices. Compared with computing at the MEC server, computing at the mobile VR device can reduce at least half of the traffic load, since the data size of the output, i.e., 3D field of view (FOV), is at least twice larger than that of the input, i.e., 2D FOV. However, computing at the mobile VR device may incur longer latency, since the computing capability of the mobile VR device is generally weaker than that of the MEC server. Thus, \textit{the computing policy}, i.e., the decision of computing the projection at the MEC server or at the mobile VR device, requires careful design. In addition, caching capability of each mobile VR device can be utilized to store some input or output data for future requests, reducing the bandwidth cost. Specifically,  compared with caching the input data of some task,
caching the output data can help reduce both latency and energy consumption, since the VR video request can be served directly from local caching and with no need of computing. However, output caching consumes larger caching resource at the mobile VR device, since the output data size is at least twice larger than the input data size. Thus, \textit{the caching policy}, i.e., the decision of caching the input or output  data at the mobile VR device, requires careful design. %for the projection computing process, a mobile VR device can download the 2D field of view (FOV) of VR video (i.e., input of the projection) from the MEC server, and then computes it into 3D FOV of VR video (i.e., output of the projection) \cite{Sunvr}. Considering that the data size of 3D FOV is at least twice larger than that of 2D FOV, computing at the mobile VR device can reduce at least half of the traffic load on the wireless link incurred by computing at the MEC server \cite{Sunvr}.
Such system model can also be commonly seen in other communication-intensive, computation-intensive and delay-sensitive applications, such as online gaming and AR \cite{E}.

%\textcolor{blue}{Therefore, a joint caching and computing policy for the mobile devices, corresponding to decisions on , determines the required bandwidth and requires careful design.} 
In general, a joint caching and computing policy for the mobile devices is to decide what tasks to cache at each mobile device, whether to cache the input or output data of each task, and whether to compute each task locally or at the MEC server. In this paper, we aim at optimizing the joint caching and computing policy such that the bandwidth cost of the wireless multicast channel is minimized and thereby illustrating the impacts of local caching and computing at mobile devices as well as content-centric multicast transmission on the saving of required bandwidth on the wireless link.
 %investigate the impact of computing and caching resources at the mobile devices in wireless multicast system on the bandwidth requirement in this work.} In this paper, by taking full advantage of such abundance of computing power and caching space in the mobile device, we present a novel mobile edge computing (MEC) architecture over wireless multicast channels to obtain significant gain in bandwidth saving.

\subsection{Contribution}
This paper presents a novel MEC architecture, which consists of a single MEC server and $K$ computing-enabled and caching-aided mobile devices, as shown in Fig.~\ref{model}. \textcolor{black}{The MEC server has the input and output data of all computation tasks and communicates with the mobile devices via a wireless multicast channel. Each mobile device can pre-store the input or output data of a task and also execute a task locally.} %The mobile device can pre-cache some of the popular tasks into the storage during the off-peak time. 
Each requested task can be served via the following four different ways: i) local output caching if the output data has already been cached locally; ii) local computing with local input caching if the input data has already been cached locally and is also chosen to be computed locally; iii) local computing without local caching if the input data is downloaded from the MEC server and then computed locally; iv) MEC computing if the output data is directly downloaded from the MEC server. As mentioned above, different caching and computing policies have different bandwidth requirements on the wireless multicast channel and therefore require careful design. %requested and not pre-cached task can be delivered to the mobile device from the MEC server by the multicast transmission, and then be computed in the mobile device. 
%Consider that each task has different sizes between the input data and the output data, e.g., the input data may be \textcolor{blue}{smaller than the output data}. This makes the location of computing be an important consideration for bandwidth cost, because it determines the transmission data size. 
% %The objective of this paper is to \textbf{optimize the joint caching and computing policy such that the bandwidth cost of the wireless multicast channels is minimized.} 
\textcolor{black}{In this paper, we mainly address the following three important issues. }
\begin{figure}[t]
\begin{center}
 \includegraphics[width=8.5cm]{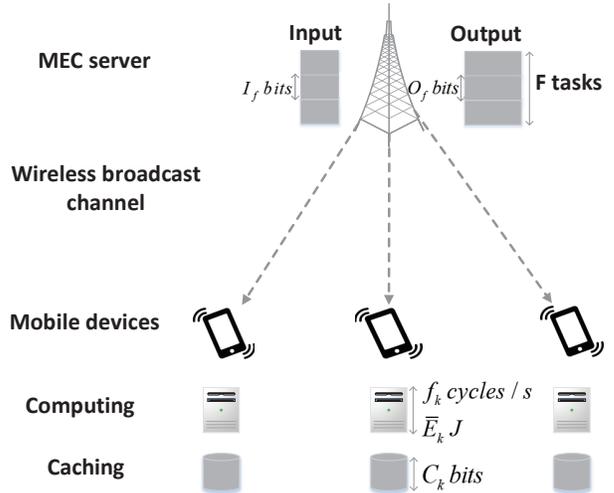}
\end{center}
 \caption{\small{A MEC system consisting of one MEC server and $K$ computing-and-caching-aided mobile devices.}
 }
\label{model}
\end{figure}

1) \textbf{What is the optimal caching and computing policy?}

In order to address this issue, we formulate the joint caching and computing policy optimization problem to minimize the average bandwidth requirement subject to the latency and multicast transmission constraints for each task \textcolor{black}{as well as} the energy and caching size constraints for each mobile device. We show that the optimization problem is a $0$-$1$ integer-programming problem, which is NP-hard in the strong sense. To tackle the problem of intractability of closed-form expression for the average bandwidth requirement since the expectation is taken over the system request space, we approximate the expectation in the objective  via sampling \cite{sample} and \cite{liuyafeng}. 

According to the relationship between the input data size and the output data size, we solve the problem respectively. \emph{For the scenario that the output data size is smaller than the input data size}, we theoretically analyze the optimal structural property and then reformulate the problem as minimization of a monotone submodular function over matroid constraints. This structure allows us to use a strongly polynomial  algorithm of Schrijver to obtain the optimal solution \cite{S}. \emph{For the scenario that the output data size is larger than the input data size}, we analyze the computation complexity and propose a low-complexity high-performance algorithm based on concave convex procedure (CCCP) in conjunction with the alternating direction method of multipliers (ADMM). In particular, firstly, in order to deal with the non-smoothness of the original problem, we reformulate the problem as a difference of convex (DC) problem \cite{exactpenalty} and \cite{infeasible}. Secondly, in order to deal with the non-convexity of the DC problem, we utilize CCCP algorithm to solve a sequence of convex subproblems. Thirdly, to reduce the computation complexity of each subproblem, we reformulate each subproblem as a consensus ADMM form, which enables that each updating step is performed by solving multiple small-size subproblems with closed-form solutions in parallel \cite{erkai,admm,admm2}.

2) \textbf{How much bandwidth reduction can be achieved by enabling caching and computing resources at mobile devices locally compared with the traditional MEC computing?}

In this issue, we try to understand the impacts of the caching and computing resources on the bandwidth gain. Unfortunately, we cannot obtain the  closed-form expression for the minimum bandwidth requirement in the general case. In the symmetric case where all the computation tasks are of the same input and output data size as well as computation load, and user requests are \textcolor{black}{uniformly distributed}, we address this issue theoretically. For example, we reveal that when $f_1 \geq \sqrt{\frac{F\bar{E}}{\mu wC}}$, % when the size of
the ratio of the minimum bandwidth requirement $B^*$ of the proposed system to that of the traditional MEC $B_{MEC}^*$ is
\begin{equation}
\frac{B^*}{B_{MEC}^*} = 1 - \beta_c - \max\left\{\left(1-\frac{1}{\alpha}\right)\beta_e,0\right\},\nonumber
\end{equation}
where $f_1$, $\bar{E}$, $w$ and $C$ are the CPU-cycle frequency (\emph{in cycles/s}), energy (\emph{in J}), computation load (\emph{in cycles/bit}) and caching size (\emph{in bits}) of each mobile device, respectively, $F$ \textcolor{black}{is} the total number of tasks, $\mu$ is a constant related to the hardware architecture, $\beta_{c} \triangleq \frac{C}{FO} \leq 1$ representing the normalized cache size at each device with respect to the total output data size of all the tasks and $\beta_{e} \triangleq \frac{\bar{E}}{\mu Iwf_1^2}\leq 1$ representing the normalized energy at each device with respect to the total average energy \textcolor{black}{consumption} of all the tasks. Here, $I$ and $O$ denote the input data size and output data size of each task, respectively, \textcolor{black}{and then} we define \textcolor{black}{$\alpha=\frac{O}{I}$}. 

Our analysis reveals that when the size of output data is smaller than that of input data, i.e., $\alpha \leq 1$, the bandwidth gain only depends on local caching but not local computing. Otherwise, the gain depends on both local caching and computing. \textcolor{black}{These analytical results offer useful guidelines for designing practical MEC-based multiuser wireless networks.}
%\begin{equation}
%\frac{B_{MEC}^*}{B^*} = \frac{F}{F-\frac{C}{O}} \nonumber,
%\end{equation}
%implying that only caching at mobile devices can bring gain. Otherwise, the gain depends on both local caching and computing, e.g.,
%\begin{equation}
%\frac{B_{MEC}^*}{B^*} = \frac{F}{F-\frac{C}{O}-(\alpha -1) \frac{F\bar{E}}{\mu Owf_1^2}},\nonumber
%\end{equation}
%when .
%Here,  $\alpha \triangleq \frac{O}{I}$  and $\mu$ is a constant related to the hardware architecture.
% i.e., $I_f=I$, $w_f = w$, $O_f=O$, $C_k=C$, $f_k=f_1$ and $\bar{E}_k=\bar{E}$ for all $f\in\mathcal{F}$ and $k\in \mathcal{K}$, we address the following two fundamental questions:

3) \textbf{How much bandwidth reduction can be achieved by exploiting multicast transmission compared with unicast transmission?}

In the symmetric case, we further find that the ratio of the minimum bandwidth requirement $B^*$ of multicast transmission to that of unicast transmission $B_{unicast}^*$ is
\begin{equation}
\frac{B^*}{B_{unicast}^*} = \frac{F(1-(1-\frac{1}{F})^K)}{K},\nonumber
\end{equation}
implying that the gain only depends on the number of mobile devices and that of tasks, \textcolor{black}{and is independent of local caching and computing capiblities}.

%\textit{compared with MEC computing, how much gain on the bandwidth requirement on earth can the caching and computing resources at mobile devices bring? In addition, compared with unicast transmission, how much gain on bandwidth can the multicast transmission bring?}

\subsection{Related Works}
%The MEC network architecture which enables caching and computing capabilities at the edge of wireless networks has been envisioned as an efficient approach to tackle the wireless spectrum crunch problem in \cite{3C,tao, erol,mohammed1}. Specifically, modern data traffic exhibits a high degree of asynchronous content reuse \cite{Caire}. Caching contents closer to users, e.g., at \textcolor{black}{base stations (BSs)} or even at end users, \textcolor{black}{can reduce traffic load, and hence} is gradually recognized as a promising approach to improve bandwidth utilization. In addition, modern data traffic also incurs heavy computation tasks as illustrated in the motivating example. Computing the tasks closer to users, e.g., at the BSs or even at end users, can reduce the required transmission delay as well as traffic, thereby reducing the required bandwidth.

\textit{Communications and Caching Model}. Caching at the MEC networks has been exploited to reduce the required transmission rate \cite{Ali, globecom, xufan,Push,chen,multicast}. \cite{Ali,globecom,xufan} design the optimal caching policy at the end-users. For example, the core idea of \cite{Ali} is how to design the cache placement and coded delivery scheme to achieve global caching gain and minimize the peak transmission rate. \cite{globecom} designs the optimal joint pushing and caching policy to maximize the bandwidth utilization and smooth the traffic load.  \cite{xufan}  studies the fundamental tradeoff between storage and latency in a general wireless interference network with caches equipped at all \textcolor{black}{the} transmitters and receivers. \cite{Push,chen,multicast} design the caching policy at the \textcolor{black}{base stations}.
However, most existing literature on caching does not exploit the computing resources at the MEC network.

\textit{Communications and Computing Model}. %Our considered MEC model differs from the traditional MEC models \cite{mao}.
In the traditional MEC model, each mobile device generates its own computation task, then decides whether to execute the task locally or offload the task to the MEC server via uplink transmission \cite{mao}. In the latter case, the MEC server needs to send the output of the computation task back to the mobile device via downlink transmission. This traditional model mainly focuses on the cost of sending the input data of each computation task in the uplink while ignoring the cost of sending back the output data in the downlink. It is thus not suitable for applications that are bandwidth hungry in downloading the computation output, such as VR video streaming. In addition, \textcolor{black}{most existing literature on computing} does not exploit the caching resources at the MEC networks. %Our considered setting looks very similar to that in \cite{Ali} which studies the fundamental limits of caching, but it is very different. The core idea of \cite{Ali} is how to design the cache placement and coded delivery scheme to achieve global caching gain, which does not exploit the computing resources at mobile devices. As shown in \cite{3C}, computing is one of the three primary resources of mobile systems, and thus this paper is not a simple extension of \cite{Ali}. On the other hand, the proposed system is also different from traditional MEC systems \cite{mao}. In the traditional MEC system, mobile devices offload tasks to the MEC server to reduce latency or local energy consumption, and then the MEC delivers the computation results to mobile devices after computation. However, this approach generally increases the communication resource consumption, and thus it may be not always suitable for high bandwidth consumption application, e.g., VR video streaming illustrated in the motivating example. In addition, taking advantage of computation to reduce the communication load is studied in \cite{song} by using coded caching method \cite{Ali} in a distributed computing system. However, the computation load in \cite{song} is defined as the average number of nodes to compute one function, similar to the caching concept in \cite{Ali}.

%Joint caching and computing at mobile devices for VR delivery has been studied in \cite{Yang} and our previous work \cite{Sunvr}. \cite{Yang} exploits the caching and computing resources at the mobile device to minimize the traffic load over wireless link. \cite{Sunvr} obtains the closed-form expression of the minimum average transmission rate, and analytically illustrates the tradeoff among communication, computing and caching. Note that \cite{Yang} and \cite{Sunvr} only consider a single-user setting and can not be easily extended to the multi-user case.

\textit{Communications, Computing and Caching Model}. Caching and computing resources at the MEC networks have been exploited  \textcolor{black}{collaboratively} in \cite{J,colla,bigdata,game,cui}. Specifically, \cite{colla} proposes a \textcolor{black}{joint} cache allocation and computation offloading policy to maximize the resource utilization in a collaborative MEC network. \cite{bigdata} extends the results in \cite{colla} to a big data MEC network. \cite{game} proposes hybrid control algorithms in smart base stations  along with devised communication, caching, and computing techniques based on game theory to maximize network resource utilization and maximize the users' QoE.  \cite{J} formulates an optimization framework for VR video delivery in a cache-enabled cooperative multi-cell network to maximize the service rewards and explores the fundamental tradeoffs between caching, computing and communication. \cite{cui} designs the optimal caching and computing offloading policy to minimize the average energy \textcolor{black}{consumption}.%\cite{mohammed2} proposes joint  policy based on millimeter wave communication for interactive VR game applications. %reveal the content-centric multicast gain for the bandwidth requirement.

It is worthy to note that  \cite{colla,bigdata,game,J,cui}  mainly try to utilize the caching and computing resources at the MEC servers to alleviate the computation burdens at the mobile devices. However, as mentioned above, taking the mobile VR delivery as an example, computing at the MEC server may incur more transmission data since the computation results are generally larger than the inputs. Joint caching and computing at mobile devices for VR delivery has been studied in \cite{Yang} and our previous work \cite{Sunvr}. \cite{Yang} exploits the caching and computing resources at the mobile device to minimize the traffic load over wireless link. \cite{Sunvr} obtains the closed-form expression of the minimum average transmission rate, and analytically illustrates the tradeoff among communication, computing and caching. Note that \cite{Yang} and \cite{Sunvr} only consider a single-user setting and can not be easily extended to the multi-user case which considers multicast.
%In this paper, we focus on utilizing the caching and computing capabilities at the mobile devices as well as content-centric multicast to alleviate the communication burden on the wireless link and minimize the bandwidth requirement.%, which is of fundamental importance to mobile carriers. %\cite{xiaoyang} exploits the caching and computing capabilities at the mobile VR device to minimize the traffic load over wireless link.

%\textbf{\textit{Paper Organizations.}}
The remainder of this paper is organized as follows. Section II introduces the considered MEC system model. In Section III, we formulate the joint caching and computing policy optimization problem, and show its NP-hardness. Section IV provides the optimal joint policy in scenarios of $\alpha \leq 1$ and $\alpha > 1$. In Section V, we analytically quantify the bandwidth gain from caching, computing and multicast. Comprehensive numerical results are provided in Section VI. Finally, conclusions are drawn in Section VII.
\section{System Model}
As illustrated in Fig.~\ref{model}, we consider a multi-user MEC system consisting of one single-antenna MEC server and $K$ single-antenna mobile devices. \textcolor{black}{It is assumed that the MEC server has access to the input and output datas of all the tasks.\footnote{It is assumed that the main required input data of each computation task is not generated by mobile devices but is available in advance at the MEC server.}  Each mobile device is endowed with a local cache with finite storage size and a local computing server with finite average energy and limited computation frequency. Each mobile device thereby can pre-store the input or output data of a task and also execute a task locally. The system operates in a time-slotted manner with each slot long enough to complete all the computation tasks. At the beginning of every time slot, each mobile device uploads a negligible amount of information to the MEC server via uplink to trigger a computation task according to certain demand probabilities and then downloads the desired data (either the input or output data) via downlink. We assume that each mobile device requests a single task at a time and each request must be served within $\tau$ seconds. Users requesting the same data (either input or output data of the same task) are grouped together and served using multicast transmission \cite{multicast}.}

\subsection{Task and Request Models}
Each task $f\in \mathcal{F}\triangleq \{1,2,\cdots,F\}$ is characterized with  input size $I_f$ (\emph{in bits}), computation load $w_f$ \textcolor{black}{(\emph{in cycles/bit})} and output size $O_f$ (\emph{in bits}) with the ratio $\alpha=\frac{O_f}{I_f}$ for all $f\in \mathcal{F}$. The task request stream at each mobile device is assumed to conform to independent reference model (IRM) based on the following assumptions: \emph{i}) the required tasks are fixed to the set $\mathcal{F}$; \emph{ii}) the probability of the request for task $f$ at mobile device $k$, denoted as $P_{k,f}$, is constant and independent of all the past requests. We have $\sum_{f\in \mathcal{F}} P_{k,f} = 1$, for all $k\in \mathcal{K}\triangleq \{1,2,\cdots,K\}$. Denote with $A_k \in \mathcal{F}$ the task requested by mobile device $k$, and $\mathbf{A} \triangleq (A_k)_{k\in \mathcal{K}} \in \mathcal{F}^K$ the system task request state, where $\mathcal{F}^K$ represents the system task request space. We assume that the $K$ task request processes are independent of each other, and thus we have $P(\textbf{A}) = \prod_{k\in \mathcal{K}} P_{k,A_k}$.

In addition, we assume that each task request must be satisfied within a given time deadline of $\tau$ seconds for quality of experience. For example, in VR video streaming, $\tau\! \approx\! 20$\!~ms to avoid dizziness and nausea \cite{E}.
%Furthermore, we assume that the mobile user requests each viewpoint with equal stationary probability \cite{zhi,J}, i.e., $P_i = 1/N$ for each $i \in \mathcal{I}$. \footnote{We take into account the diversity among the viewpoints in Section~\ref{heterogeneous}.}%. Specifically, each $i\in \mathcal{I}$ is characterized by $(I_i,O_i,w_i,\tau_i)$, respectively.

\subsection{Caching and Computing \textcolor{black}{Models}}
Each mobile device $k$ %is endowed with limited computing \textcolor{blue}{(frequency $f_k$ (\emph{in cycles/s}), average energy $\bar{E}_k$ (\emph{in J}))} and caching $C_k$ \textcolor{blue}{(\emph{in bits})} resources, and 
can trigger a computation task from $\mathcal{F}$ randomly at each time. First, consider the cache placement at mobile device $k$, for all $k\in \mathcal{K}$. We consider that each mobile device $k$ is equipped with a cache size $C_k$ \textcolor{black}{(\emph{in bits})}, and is able to store the input or output data of some tasks.
%the mobile VR is equipped with a cache size $C$ (\emph{in bits}) and able to store the 2D FOV for each viewpoint $i \in \mathcal{I}$.
Denote with $c_{k,f}^I \in \{0,1\}$ the caching decision for input data of task $f$, where $c_{k,f}^I = 1$ means that the input data of task $f$ is cached in the mobile device $k$, and $c_{k,f}^I = 0$ otherwise. Denote with $c_{k,f}^O \in \{0,1\}$ the caching decision for output data of task $f$, where $c_{k,f}^O = 1$ means that the output data of task $f$ is cached in the mobile device $k$, and $c_{k,f}^O = 0$ otherwise. %$0<c_i<1$ means that a portion of the 2D FOV for task $f$ is cached and $c_i = 0$ means that the 2D FOV for task $f$ is not cached in the mobile VR.
Under the cache size constraint, we have
\begin{equation}\label{CacheSize}
\sum_{f=1}^F I_fc_{k,f}^I + O_f c_{k,f}^O \leq C_k,\ k\in \mathcal{K}.
\end{equation}
Denote with $(\textbf{c}^I,\textbf{c}^O)$ the system caching decision, where $\textbf{c}^I \triangleq (c_{k,f}^I)_{k\in \mathcal{K}, f\in \mathcal{F}}$ and $\textbf{c}^O \triangleq (c_{k,f}^O)_{k\in \mathcal{K}, f\in \mathcal{F}}$ satisfy the cache size constraint in (\ref{CacheSize}).
%Here, we assume that $C< N$ for the interest of caching design. In addition, we assume that the 2D FOVs of all the viewpoints are extracted from the 2D plane video at the cloud server offline and then are already cached at the MEC server,  since the storage size at the MEC server is much larger than that of mobile device.
%where $C \triangleq \frac{C}{I}$ represents the maximum number of 2D FOVs that can be cached at the mobile VR, referred to as caching capacity of the mobile VR.\footnote{For convenience, we assume that $C\hspace{-0.0002mm}\mod\! I =\! 0$ in this paper.}%Without loss of generality, we assume $C\!\!\!\mod \!I =0$ in this paper.

Next, consider the computing decision at mobile device $k$, for all $k\in \mathcal{K}$. Each mobile device $k$ is equipped with a computing server, which can run at a constant CPU-cycle frequency $f_k$ \textcolor{black}{(\emph{in cycles/s})} and with a fixed average energy $\bar{E}_k$ \textcolor{black}{(\emph{in J})}. The power consumed at the mobile device for computation per cycle  with frequency $f_k$ is $\mu f_k^3$. Denote with $d_{k,f} \in \{0,1\}$ the computation decision for task $f$, where $d_{k,f}= 1$ means that task $f$ is computed at the mobile device $k$, and $d_{k,f} = 0$ otherwise. Under the average energy consumption constraint, we have
\begin{equation}\label{energy}
\sum_{f=1}^F P_{k,f}\mu f_k^2 I_fw_f d_{k,f} \leq \bar{E}_k,\ k\in \mathcal{K}.
\end{equation}
%where $\bar{E}_k \triangleq W_k\sum_{f=1}^F P_{k,f}\frac{I_fw_f}{f_k}$ (\emph{in J}) represents the average energy available at mobile device $k$ during local computing.
Denote with $\textbf{d} \triangleq (d_{k,f})_{k\in \mathcal{K},f\in \mathcal{F}}$ the system computing decision, which satisfies the average energy consumption constraint in (\ref{energy}).
% the projection component is computed at the MEC server and $d_i = 1$ means that the projection component is computed at the mobile device. %Denote with $R_i$ (in bit/s) the transmission rate for satisfying the request for the 3D FOV of task $f$.

\subsection{Service Mechanism}

Based on the joint caching and computing decision, i.e., $(\textbf{c}^I,\textbf{c}^O,\textbf{d})$, we can see that request for task $f \in \mathcal{F}$ at mobile device $k \in \mathcal{K}$ can be served via the following four routes, each of which yields a unique transmission rate requirement. Denote with $R_{f,j}^k$ (\emph{in bits/s}) the minimum transmission rate required for satisfying task $f$ at mobile device $k$ via Route~$j \in \{1,2,3,4\}$ within the deadline $\tau$ seconds.

\begin{itemize}
\item \textit{\textbf{Route~1: Local output caching}}. If $c_{k,f}^O=1$, i.e., the output data of task $f$ has been cached at the mobile device $k$, request for task $f$ can be satisfied directly from the cache of mobile device $k$, thereby without any need of computing or transmission. Thus, the required latency is negligible and $R_{f,1}^k = 0$.
\item \textit{\textbf{Route~2: Local computing with local input caching}}. If $c_{k,f}^O=0$, but $c_{k,f}^I=1$ and $d_{k,f} = 1$, i.e., the input data of task $f$ has been cached and computed at the mobile device $k$, request for task $f$ can be satisfied via local computing based on the cached input data, thereby without any need of transmission. Thus, the required latency is $\frac{I_fw_f}{f_k}$ and $R_{f,2}^k = 0$. For feasibility, we assume that $\frac{I_fw_f}{f_k} \leq \tau, f\in \mathcal{F}, k\in \mathcal{K}$.
\item \textit{\textbf{Route~3: Local computing without local caching}}. If $c_{k,f}^O=0$, $c_{k,f}^I=0$ and $d_{k,f} = 1$, i.e., the output or input data of task $f$ has not been cached and task $f$ is chosen to be computed at the mobile device $k$, the execution for satisfying task $f$ consists of the following two stages: \emph{i}) the input data of task $f$ is transmitted from the MEC server; \emph{ii}) the input data is computed at the mobile device $k$. Thus, the required latency is $\frac{I_f}{R_{f,3}^k} + \frac{I_fw_f}{f_k}$. Under the latency constraint, we have $\frac{I_f}{R_{f,3}^k} + \frac{I_fw_f}{f_k} = \tau$, i.e., $R_{f,3}^k = \frac{I_f}{\tau-\frac{I_fw_f}{f_k}}$.
\item \textit{\textbf{Route~4: MEC computing}}. If $c_{k,f}^O = 0$, $c_{k,f}^I = 0$ and $d_{k,f} = 0$, i.e., output or input data of task $f$ has not been cached and task $f$ is not chosen to be computed locally, task $f$ is satisfied via downloading the output data from the MEC server. Thus, the required latency is $\frac{O_f}{R_{f,4}^k}$. Under latency constraint, we have $\frac{O_f}{R_{f,4}^k} = \tau$, i.e., $R_{f,4}^k = \frac{O_f}{\tau}$.
\end{itemize}

In summary, denote with $x_{f,j}^k\in \{0,1\}$ the service decision for task $f$ at mobile device $k$, where $x_{f,j}^k = 1$ means that task $f$ at mobile device $k$ is served via \textcolor{black}{above-mentioned} Route~$j\in \{1,2,3,4\}$, and $x_{f,j}^k =0$ otherwise. To guarantee that task $f$ at mobile device $k$ gets served, we have
\begin{equation}\label{schedule}
\sum_{j=1}^4 x_{f,j}^k = 1,\ f\in \mathcal{F},\ k\in \mathcal{K}.
\end{equation}

In addition, the cache size and average energy consumption constraints in (\ref{CacheSize}) and (\ref{energy}) can be rewritten as
\begin{equation}\label{cachesize2}
\sum_{f=1}^F I_f x_{f,2}^k + O_f x_{f,1}^k \leq C_k,\ k\in \mathcal{K},
\end{equation}
\begin{equation}\label{energy2}
\sum_{f=1}^F P_{k,f}\mu f_k^2I_fw_f(x_{f,2}^k+x_{f,3}^k) \leq \bar{E}_k,\ k\in \mathcal{K}.
\end{equation}
\begin{table}[t]
\caption{Transmission Rates vs. Local Caching and Computing Costs}\label{tradeoff}
\newcommand{\tabincell}[2]{\begin{tabular}{@{}#1@{}}#2\end{tabular}}
\begin{center}
\begin{tabular}{lccc}
\hline
\ \ \ \ \ \ \ \ \ \ Service\! Route  & Rate  & Caching\! Cost & Computing\! Cost \\
\hline
\tabincell{c}{$x_{f,1}^k\!=\!1$\\($c_{k,f}^O\!=\!1,\!c_{k,f}^I\!=0,\! d_{k,f}\!=\!0$)}  & $0$ & $O_f$ & $0$\\
\hline
\tabincell{c}{$x_{f,2}^k\!=\!1$\\($c_{k,f}^O\!=\!0,\!c_{k,f}^I\!=1, \!d_{k,f}\!=\!1$)} & $0$ & $I_f$  & $P_{k,f}\mu I_fw_ff_k^2$\\
\hline
\tabincell{c}{$x_{f,3}^k\!=\!1$\\($c_{k,f}^O\!=\!0,\!c_{k,f}^I\!=0, \!d_{k,f}\!=\!1$)}   & $R_{f,3}^k$ & $0$ & $P_{k,f}\mu I_fw_ff_k^2$ \\
\hline
\tabincell{c}{$x_{f,4}^k\!=\!1$\\($c_{k,f}^O\!=\!0,\!c_{k,f}^I\!=0, \!d_{k,f}\!=\!0$)}  & $R_{f,4}^k$ & $0$ & $0$\\
\hline
\end{tabular}
\end{center}
\label{a}
\end{table}%

For clarity, we illustrate the relationship between the service policy $\textbf{x} \triangleq (x_{f,j}^k)_{f\in \mathcal{F}, j\in \{1,2,3,4\},k\in \mathcal{K}}$ and joint caching and computing policy, i.e., $(\textbf{c}^I,\textbf{c}^O,\textbf{d})$, as well as the associated transmission rates and local caching and computing costs in Table~I.

\subsection{Multicast Transmission Model}
At each time slot, given system task request state $\textbf{A}$ and service decision $\textbf{x}$, \textcolor{black}{users requesting the same data (either input or output data of the same task) are grouped together and served using multicast transmission.} %\cite{multicast}. %the MEC server employs multicasting to serve multiple mobile devices concurrently if they request the same task. %simultaneously serve many different requests for either input or output data of the same task.
Specifically, denote with $B_f^I(\textbf{x}, \textbf{A})$ and $B_f^O(\textbf{x},\textbf{A})$ (in \emph{Hz}) the bandwidth allocated by the MEC server for transmitting the input and output data of task $f\in \mathcal{F}$, respectively. To guarantee each user's QoE and  considering that the multicast rate is limited by the user with the worst channel condition, we have
\begin{align}\label{input}
&B_f^I(\textbf{x},\textbf{A})\min_{k\in \mathcal{K}} \log\left(1+\frac{Ph_k^2}{\sigma^2}\right)\textbf{1}(A_k=f)x_{f,3}^k\nonumber\\
&\ \ \ \ \ \ \ \ \ \ \ \ \  \ \ \ \ \ \ \ \ \geq \max_{k\in \mathcal{K}} R_{f,3}^k\textbf{1}(A_k=f)x_{f,3}^k,\ f\in \mathcal{F},
\end{align}
\begin{align}\label{output}
&B_f^O(\textbf{x},\textbf{A})\min_{k\in \mathcal{K}} \log\left(1+\frac{Ph_k^2}{\sigma^2}\right)\textbf{1}(A_k=f)x_{f,4}^k\nonumber\\
&\ \ \ \ \ \ \ \ \ \ \ \ \  \ \ \ \ \ \ \ \ \geq \max_{k\in \mathcal{K}} R_{f,4}^k\textbf{1}(A_k=f)x_{f,4}^k,\ f\in \mathcal{F},
\end{align}
where $P$ denotes the transmission power of the MEC server, $\sigma^2$ denotes the variance of complex white Gaussian channel noise, and $h_k$ denotes the channel gain for mobile device $k$, which is assumed to be constant within the deadline $\tau$ seconds, respectively. $\textbf{1}(\cdot)$ denotes the indicator function throughout the paper.

%Given $\textbf{A}$, denote with $R_f(\textbf{x},\textbf{A})$ (\emph{in bits/s}) the transmission rate for satisfying task $f$ under $\textbf{x}$, and we have
%\begin{align}
%&R_f(\textbf{x}, \textbf{A}) = \max_{k\in \mathcal{K}} R_{f,3}^k\textbf{1}(A_k=f)x_{f,3}^k \nonumber\\
%&\hspace{20mm}+ \max_{k\in \mathcal{K}} R_{f,4}^k \textbf{1}(A_k=f)x_{f,4}^k,\ f\in \mathcal{F}.
%\end{align}
Under $\textbf{x}$, denote with $B(\textbf{x})$ the average bandwidth requirement, and we have
\begin{equation}\label{averagerate}
B(\textbf{x}) = \mathbb{E}\left[ \left(\sum_{f=1}^F B^I_f(\textbf{x},\textbf{A}) + B^O_f(\textbf{x},\textbf{A})\right)\right],
\end{equation}
where the expectation is taken over the system request state $\textbf{A}\! \in\! \mathcal{F}^K$.
%\subsection{Service Model and Transmission Rate Requirement}

%To satisfy the request for the 3D FOV of task $f$, the mobile device first checks its own cache placement and then decides whether to execute the projection component at the MEC server or the mobile device. %for task $f$.

%The MEC server is usually powered by an electricity grid and thus we do not take into account the energy consumption at the MEC server. On the other hand, the mobile device is powered by its own battery.
%where $ NE \triangleq \frac{E}{k f_1^2Iw} $ represents the maximum number of projections that can be computed at the mobile VR, referred to as computation capacity of the mobile VR.\footnote{In this paper, we assume that $E\hspace{-0.2mm} \mod\! {k f_1^2Iw} = 0$.}
%Thus, the number of VR tasks that can be offloaded to the mobile VR, denoted as $R \triangleq \sum_{i=1}^N d_i$, is no larger than $$
\section{Problem Formulation and Analysis}
In this paper, our objective is to minimize the average bandwidth requirement subject to the latency,  \textcolor{black}{multicast transmission}, cache size and average energy consumption constraints. The optimization problem can be formulated as the following $0$-$1$ integer-programming problem.
\begin{Prob}[Average Bandwidth Minimization]\label{Prob1}
\begin{align}
& \min_{\textbf{x}} \ \ \ \ \ \ \ \ \ \ \ \ \ \ \ \ \ \ \ \ \ \ \ \ \ \ B(\textbf{x}) \nonumber\\
& \ s.t. \ \ \ \ \ \ \ \ \ \ \ \ \ \ \ \ \ \ \ \ (\ref{schedule}),(\ref{cachesize2}),(\ref{energy2}),(\ref{input}),(\ref{output}), \nonumber\\
&\ \ \ \ \ \ x_{f,j}^k \in \{0,1\},\ f\in \mathcal{F},\ k\in \mathcal{K},\ j\in \{1,2,3,4\}.\label{binary1}
\end{align}
\end{Prob}
Denote with $B^*$ the minimum average bandwidth, and $\textbf{x}^*$ the optimal service decision. Thus, we have $B^* = B(\textbf{x}^*)$ and then obtain the optimal joint caching and computing policy $(\textbf{c}^{I*},\textbf{c}^{O*}, \textbf{d}^*)$ from Table~I with $\textbf{x}^*$.

%In order to minimize the average bandwidth requirement,
It is direct to observe that (\ref{input}) and (\ref{output}) are reduced to equality for optimality, and accordingly we have
\begin{align}\label{input1}
&B_f^I(\textbf{x},\textbf{A}) = \max_{k\in \mathcal{K}} \frac{1}{\log\left(1+\frac{Ph_k^2}{\sigma^2}\right)}\textbf{1}(A_k=f)x_{f,3}^k\nonumber\\
&\ \ \ \ \ \ \ \ \ \ \ \ \  \ \ \ \ \ \  \times\max_{k\in \mathcal{K}} R_{f,3}^k\textbf{1}(A_k=f)x_{f,3}^k,\ f\in \mathcal{F},
\end{align}
\begin{align}\label{output1}
&B_f^O(\textbf{x},\textbf{A}) = \frac{O_f}{\tau}\max_{k\in \mathcal{K}} \frac{1}{\log\left(1+\frac{Ph_k^2}{\sigma^2}\right)}\textbf{1}(A_k=f)x_{f,4}^k,\nonumber\\
&\ \ \ \ \ \ \ \ \ \ \ \ \  \ \ \ \ \ \ \ \ \ \ \ \ \ \ \ \ \ \ \ \ \ \ \ \ \ \ \ \ \ \ \ \ \ \ \ \ \ \ \ \ f\in \mathcal{F}.
\end{align}
\subsection{Computation Intractability}
In the following, we show that Problem~\ref{Prob1} is NP-hard in the strong sense. Consider a single user scenario for the problem, i.e., $K=1$. %At the single mobile device,
For all task $f\in \mathcal{F}$, denote with $x_{f,j}$ the service decision for Route $j \in \{1,2,3,4\}$, $P_f$ the request probability, $C$ the cache size, $\bar{E}$ the average energy, $f_1$ the computation frequency (\emph{in cycles/s}) and $h$ the channel gain of the single mobile device.  In this case, Problem~\ref{Prob1} can be formulated as the following maximization problem:
\begin{Prob}[Optimization for Single User Scenario]\label{mmkp}
\begin{align}
& \max_{(x_{f,j})_{f\in \mathcal{F}, j\in \{1,2,3,4\}}} \ \ \ \ \ \ \  \ \ \ \ \sum_{f = 1}^F \sum_{j=1}^4 v_{f,j}x_{f,j}\nonumber \\
&\ \ \ \ \ \ \ \ s.t. \ \ \ \ \ \ \  \ \ \ \ \ \ \ \ \ \sum_{f=1}^F\sum_{j=1}^4 w_{f,j}^1x_{f,j} \leq C, \label{cache1}\\
&\ \ \ \ \ \ \  \ \ \ \ \ \ \ \ \ \ \ \ \ \ \ \ \ \ \ \ \ \sum_{f=1}^F\sum_{j=1}^4 w_{f,j}^2x_{f,j} \leq \bar{E},\label{energy22}\\
&\ \ \ \ \ \ \  \ \ \ \ \ \ \ \ \ \ \ \ \ \ \ \ \ \ \ \ \ \sum_{j=1}^4 x_{f,j} = 1, \ f \in \mathcal{F},\label{sum}\\
&\ \ \ \ \ \ \  \ \ \ \ \ \ \ \ \  x_{f,j} \in \{0,1\}, \ f \in \mathcal{F},\ j \in \{1,2,3,4\}\label{binary},
\end{align}
\end{Prob}
where
\begin{equation}\label{value}
v_{f,j} \triangleq
\begin{cases}
P_f\frac{O_f}{\tau}\frac{1}{\log(1+\frac{Ph^2}{\sigma^2})},& \ \ \ \ \ \text{$j=1,2$,}\\
P_f\left(\frac{O_f}{\tau}-\frac{I_f}{\tau-\frac{I_fw_f}{f_1}}\right)\frac{1}{\log(1+\frac{Ph^2}{\sigma^2})},& \ \ \ \ \ \text{$j=3$,}\\
0,& \ \ \ \ \ \text{$j=4$,}
\end{cases}
\end{equation}
denotes the profit value for the choice of Route $j$ for task $f$,
\begin{equation}\label{value}
w_{f,j}^1 \triangleq
\begin{cases}
O_f,& \ \ \ \ \ \text{$j=1$,}\\
I_f, & \ \ \ \ \ \text{$j=2$,}\\
0, & \ \ \ \ \ \text{$j=3,4$,}
\end{cases}
\end{equation}
denotes the caching cost for the choice of Route $j$ for task $f$, and
\begin{equation}\label{value}
\ \ \ \ \ \ \ \ w_{f,j}^2 \triangleq
\begin{cases}
P_f \mu I_fw_ff_1^2,& \ \ \ \ \ \text{$j=2,3$,}\\
0, & \ \ \ \ \ \text{$j=1,4$,}\\
%0 & \ \ \ \ \ \text{$k=3,4$.}
\end{cases}
\end{equation}
denotes the energy cost for the choice of Route $j$ for task $f$.

We can see that Problem~\ref{mmkp} is a $4$-choice $2$-dimensional knapsack problem, which is a well-known NP-hard problem in the strong sense \cite{mmkp}. Thus, we conclude that Problem~1 in the multiple-user scenario is also NP-hard in strong sense.

\subsection{Equivalent Problem Reformulation}
Furthermore, it is difficult to derive a closed-form expression for the objective function in (\ref{averagerate}) since the expectation is taken over the systematic request space $\mathcal{F}^K$. We replace the objective function in (\ref{averagerate}) with its sample approximation \cite{sample, liuyafeng} and  reformulate Problem~\ref{Prob1} as:
\begin{Prob}[Equivalent Problem Reformulation]\label{approximate}
\begin{align}
& \min_{\textbf{x}} \ \ \ \ \ \ \ \  \frac{1}{N} \sum_{n=1}^N \sum_{f=1}^F \left(B^I_f(\textbf{x},\textbf{A}_n) + B^O_f(\textbf{x},\textbf{A}_n)\right)\nonumber\\
& \ s.t. \ \ \ \ \ \ \ \ \ \ \ \ (\ref{schedule}),(\ref{cachesize2}),(\ref{energy2}), (\ref{binary1}), \nonumber\\%(\ref{input}),(\ref{output})
& \ \ \ \ \ \ \ B_f^I(\textbf{x},\textbf{A}_n) = \max_{k\in \mathcal{K}} \frac{1}{\log\left(1+\frac{Ph_k^2}{\sigma^2}\right)}\textbf{1}(A_{n,k}=f)x_{f,3}^k\nonumber\\
&\ \ \ \ \ \ \  \ \ \ \ \ \ \ \ \ \ \ \ \ \ \ \ \ \times \max_{k\in \mathcal{K}} R_{f,3}^k\textbf{1}(A_{n,k}=f)x_{f,3}^k,\ \nonumber\\
&\ \ \ \ \ \ \ \ \ \ \ \ \ \ \ \ \ \ \ \ \ \ \ \ \ \ \ \ \ \ \ \ \ \ \ \ \ \ \ \ \ \ f\in \mathcal{F},\ n\in \mathcal{N}, \label{sample1}\\
&\ \ \ \ \ \ \ B_f^O(\textbf{x},\textbf{A}_n) = \frac{O_f}{\tau}\max_{k\in \mathcal{K}} \frac{1}{\log\left(1+\frac{Ph_k^2}{\sigma^2}\right)}\textbf{1}(A_{n,k}=f)x_{f,4}^k\nonumber\\
%&\ \ \ \ \ \ \ \ \ \ \ \ \  \ \ \ \ \ \ \ \ \ \  *\max_{k\in \mathcal{K}} R_{f,4}^k\textbf{1}(A_{n,k}=f)x_{f,4}^k,\ \nonumber\\
&\ \ \ \ \ \ \ \ \ \ \ \ \ \ \ \ \ \ \ \ \ \ \ \ \ \ \ \ \ \ \ \ \ \ \ \ \ \ \ \ \ \ f\in \mathcal{F},\ n\in \mathcal{N}, \label{sample2}
\end{align}
where $N$ is the sample size, $\mathcal{N} \triangleq \{1,\cdots,N\}$ and $\{\textbf{A}_n\}_{n\in \mathcal{N}} \triangleq \{(A_{n,k})_{k\in \mathcal{K}}\}_{n\in \mathcal{N}}$ are the request samples drown according to the distribution of $\textbf{A}$.
\end{Prob}
\section{Optimal Policy Design}
 In this section, \textcolor{black}{when the output data size is smaller than the input data size ($\alpha \leq 1$)}, we reformulate the problem as minimization of a monotone submodular function over matroid constraints via analyzing the optimal structural property. \textcolor{black}{On the other hand, when $\alpha > 1$,} we analyze its computation complexity and propose a low-complexity high-performance algorithm named as CCCP-ADMM to obtain a stationary point.
\subsection{Optimal policy design for $\alpha \leq 1$}
When $\alpha \leq 1$, we obtain the structural properties of the optimal service policy as below.
\begin{property}\label{optimalproperty}
For all $f\in \mathcal{F}$ and $k\in \mathcal{K}$, we have $x_{f,2}^{k,*} =0$ and $x_{f,3}^{k,*}=0$, i.e., $c_{k,f}^{I,*} =0$ and $d_{k,f}^*=0$.
\end{property}

Property~\ref{optimalproperty} can be proofed by contradiction. First, suppose that there exist $k' \in \mathcal{K}$ and $f' \in \mathcal{F}$ such that $x_{f',2}^{k',*} = 1$. However, by setting $x_{f',2}^{k',*} $ from $1$ to $0$ and $x_{f',1}^{k',*} $ from $0$ to $1$, $B(\textbf{x}^*)$ does not change and caching cost is saved. Thus, $x_{f',2}^{k',*} = 1$ is not optimal.  Secondly, suppose that there exist $k' \in \mathcal{K}$ and $f' \in \mathcal{F}$ such that $x_{f',3}^{k',*} =1$. However, by setting $x_{f',3}^{k',*} $ from $1$ to $0$ and $x_{f',4}^{k',*} $ from $0$ to $1$, $B(\textbf{x}^*)$ does not increase since $R_{f',3}^{k'} \geq R_{f',4}^{k'}$ when $\alpha \leq 1$ and computing cost is saved. Thus, $x_{f',3}^{k',*} = 1$ is not optimal.

%we have $O_f \leq I_f$ and $R_{f,3}^k \geq R_{f,4}^k$ for all $f\in \mathcal{F}$ and $k\in \mathcal{K}$. Thus, from Table~I, we see that compared with MEC computing, local computing brings no transmission rate gain (i.e., $R_{f,3}^k \geq R_{f,4}^k$), but incurs additional local energy consumption (i.e., $P_{k,f}\mu I_fw_ff_k^2$); compared with local output caching, local input caching with local computing brings no transmission gain (i.e., $0$), but incurs additional local cache cost (i.e., $I_f-O_f \geq 0$) and energy resource consumption (i.e., $P_{k,f}\mu I_fw_ff_k^2$). Thus,
Property~\ref{optimalproperty} indicates that when $\alpha \leq 1$, \emph{there is no gain from local input caching and local computing, and only exists output caching gain}. Based on Property~\ref{optimalproperty}, Problem~\ref{Prob1} can be reformulated as Problem~\ref{alpha<1}.
\begin{Prob}[\textcolor{black}{Equivalent Optimization} when $\alpha \leq 1$]\label{alpha<1}
\begin{align}
&(x_{f,1}^{k,*})_{f\in \mathcal{F},k\in \mathcal{K}}\nonumber\\
&\triangleq \arg \min_{(x_{f,1}^k)_{f\in \mathcal{F},k\in \mathcal{K}}}   \sum_{n=1}^N\sum_{f=1}^F\! \frac{O_f}{\tau}\!\max_{k\in \mathcal{K}:A_{n,k}=f} \frac{1-x_{f,1}^k}{\log\left(1+\frac{Ph_k^2}{\sigma^2}\right)}\nonumber\\
&\ \ \ \ \ \ \ \ \ \ \ \ \ s.t.\ \ \ \ \ \ \ \ \  \ \sum_{f=1}^F O_f x_{f,1}^k \leq C_k,\ k\in \mathcal{K}, \nonumber \\
&\ \ \ \ \ \ \ \ \ \ \ \ \ \ \ \ \ \ \ \ \ \ \ \ \ \ \ x_{f,1}^k \in \{0,1\},\  f\in \mathcal{F},\ k\in \mathcal{K}. \nonumber
\end{align}
And $x_{f,4}^{k,*} = 1- x_{f,1}^{k,*}$, $\forall\ f\in \mathcal{F}$ and $k \in \mathcal{K}$.
\end{Prob}
\begin{lemma}\label{sub}
Problem~\ref{alpha<1} is a monotonically nonincreasing submodular function minimization problem subject to matroid constraints. 
\end{lemma}
\begin{proof}
Please see Appendix~A. 
\end{proof}
%We can see that Problem~\ref{alpha<1} is a monotonically nonincreasing submodular function minimization problem subject to matroid constraints \cite{S} and \cite{femto}. 
This structure allows us to use a strongly polynomial algorithm of Schrijver to obtain the optimal solution \cite{S}.

\subsection{Optimal policy design for $\alpha > 1$}

When $\alpha > 1$, Problem~\ref{approximate} is not easy to solve mainly due to the following three reasons. First, the objective function is nonsmooth and nonconvex. Secondly,   (\ref{binary1}) are binary constraints, albeit (\ref{schedule}), (\ref{cachesize2}) and (\ref{energy2}) are convex. Thirdly, the sample size $N$ generally needs to be large enough such that the sample average is a good approximation to the original expectation \cite{liuyafeng}.  Thus,  solving Problem~\ref{approximate} directly is of high computation complexity. In the following, we first reformulate Problem~\ref{approximate} as a  continuous smooth DC problem, and then leverage CCCP to approximate the nonconvex problem as a sequence of convex subproblems. Each convex subproblem is then reformulated as a consensus ADMM form. The ADMM reformulation enables that each updating step is performed by solving multiple small-size subproblems with closed-form solutions in parallel. Finally, we obtain a stationary point of the original problem.

%In particular,

\subsubsection{DC problem formulation} Firstly, for all $f\in \mathcal{F}$ and $n\in \mathcal{N}$, introduce auxiliary variables $a^I_{f,n}$, $b^I_{f,n}$ and $a^O_{f,n}$ satisfying
\begin{align}\label{ai}
&a^I_{f,n} = \max_{k\in \mathcal{K}} \frac{1}{\log\left(1+\frac{Ph_k^2}{\sigma^2}\right)}\textbf{1}(A_{n,k}=f)x_{f,3}^k,\nonumber\\
&\hspace{53mm} f\in \mathcal{F}, \ n \in \mathcal{N},
\end{align}
\begin{align}\label{bi}
b^I_{f,n} = \max_{k\in \mathcal{K}} R_{f,3}^k \textbf{1}(A_{n,k}=f) x_{f,3}^k,\  f \in \mathcal{F}, \ n \in \mathcal{N},
\end{align}
\begin{align}\label{ao}
&a^O_{f,n} = \max_{k\in \mathcal{K}} \frac{1}{\log\left(1+\frac{Ph_k^2}{\sigma^2}\right)}\textbf{1}(A_{n,k}=f)x_{f,4}^k,\nonumber\\
&\hspace{53mm}f\in \mathcal{F},\ n \in \mathcal{N},
\end{align}
%\begin{align}\label{bo}
%b^O_{f,n} = \max_{k\in \mathcal{K}} R_{f,4}^k \textbf{1}(A_{n,k}=f)x_{f,4}^k, \ f \in \mathcal{F}, \ n \in \mathcal{N},
%\end{align}
respectively.
 %$a^O(f,\textbf{A}) \triangleq \max_{k\in \mathcal{K}} \frac{1}{\log\left(1+\frac{Ph_k^2}{\sigma^2}\right)}\textbf{1}(A_k=f)x_{f,4}^k$, and $b^O(f,\textbf{A}) \triangleq \max_{k\in \mathcal{K}} R_{f,4}^k\textbf{1}(A_k=f)x_{f,4}^k$.
% satisfying  $a^I(f,\textbf{A}) = \max_{k\in \mathcal{K}} \frac{1}{\log\left(1+\frac{Ph_k^2}{\sigma^2}\right)}\textbf{1}(A_k=f)x_{f,3}^k$, satisfying  $b^I(f,\textbf{A}) = \max_{k\in \mathcal{K}} R_{f,3}^k\textbf{1}(A_k=f)x_{f,3}^k$, $a^O(f,\textbf{A})$ satisfying  $a^O(f,\textbf{A}) = \max_{k\in \mathcal{K}} \frac{1}{\log\left(1+\frac{Ph_k^2}{\sigma^2}\right)}\textbf{1}(A_k=f)x_{f,4}^k$,  satisfying  $b^O(f,\textbf{A}) = \max_{k\in \mathcal{K}} R_{f,4}^k\textbf{1}(A_k=f)x_{f,4}^k$.
 Accordingly, $B^I_f(\textbf{x},\textbf{A}_n)$ in (\ref{sample1}) and $B^O_f(\textbf{x},\textbf{A}_n)$ in (\ref{sample2}) can be rewritten as
\begin{align}
&B_f^I(\textbf{x},\textbf{A}_n) = \frac{\left(a^I_{f,n} + b^I_{f,n}\right)^2}{4}- \frac{\left(a^I_{f,n} \!- b^I_{f,n}\right)^2}{4},\nonumber\\
&\hspace{53mm} f\in \mathcal{F},\ n \in \mathcal{N},\label{input2}\\
&B_f^O(\textbf{x},\!\textbf{A}_n) = \frac{O_f}{\tau}a^O_{f,n},\ f\in \mathcal{F},\ n\in \mathcal{N}, \label{output2} %\frac{\left(a^O_{f,n} + b^O_{f,n}\right)^2}{4}- \frac{\left(a^O_{f,n} - b^O_{f,n}\right)^2}{4},
\end{align}
respectively, each of which is a DC function.

Secondly, (\ref{binary1}) can be rewritten as
\begin{equation}\label{aa}
x_{f,j}^k \in [0,1], \ f \in \mathcal{F},\ j \in \{1,2,3,4\},\ k\in \mathcal{K},
\end{equation}
\begin{equation}\label{b}
\sum_{k=1}^K\sum_{f=1}^F \sum_{j=1}^4 x_{f,j}^k(1-x_{f,j}^k) \leq 0.
\end{equation}

Then, by substituting  $B^I_f(\textbf{x},\textbf{A}_n)$ and $B^O_f(\textbf{x},\textbf{A}_n)$ with (\ref{input2}) and (\ref{output2}), (\ref{binary1}) with (\ref{aa}) and (\ref{b}), respectively, Problem~\ref{approximate} is reformulated as Problem~\ref{dcproblem}.

\begin{Prob}[Equivalent DC Problem]\label{dcproblem}
\begin{align}
& \min_{\textbf{x}} \ \ \ \ \ \ \ \  \frac{1}{N}\sum_{n=1}^N \sum_{f=1}^F \left(B^I_f\left(\textbf{x},\textbf{A}_n\right) + B^O_f\left(\textbf{x},\textbf{A}_n\right)\right)\nonumber\\
& \ s.t. \ \ \ \ \ \  (\ref{schedule}),(\ref{cachesize2}),(\ref{energy2}),(\ref{ai}), (\ref{bi}), (\ref{ao}), (\ref{input2}), (\ref{output2}), (\ref{aa}), (\ref{b}).\nonumber%(\ref{input}),(\ref{output})
\end{align}
\end{Prob}

Note that Problem~\ref{dcproblem} is a continuous DC problem. However, (\ref{b}) is not a convex constraint, and thus obtaining an efficient algorithm for solving Problem~\ref{dcproblem} is still very challenging.

\subsubsection{Penalized formulation and CCCP algorithm} To facilitate the problem solving, we transform Problem~\ref{dcproblem} into Problem~\ref{penalized} by penalizing the concave constraints in (\ref{b}) to the objective function.

\begin{Prob}[Penalized Optimization]\label{penalized}
\begin{align}
& \min_{\textbf{x}} \ \ \  \frac{1}{N} \sum_{n=1}^N \sum_{f=1}^F \left(B^I_f\left(\textbf{x},\textbf{A}_n\right) + B^O_f\left(\textbf{x},\textbf{A}_n\right)\right)\nonumber\\
&\ \ \ \ \ \ \ -\rho \sum_{k=1}^K\sum_{f=1}^F\sum_{j=1}^4x_{f,j}^k(x_{f,j}^k-1)\nonumber \\
& s.t. \ \ \  (\ref{schedule}), (\ref{cachesize2}), (\ref{energy2}), (\ref{ai}), (\ref{bi}), (\ref{ao}), (\ref{input2}),(\ref{output2}),(\ref{aa}), \nonumber
%&\ \ \ \ \ a^I(f,\textbf{A}) \geq \frac{1}{\log\left(1+\frac{Ph_k^2}{\sigma^2}\right)}\textbf{1}(A_k=f)x_{f,3}^k,\nonumber\\
%&\ \ \ \ \ \ \ \ \ \ \ \ \ \ \ \ \ \ \ \ \ \ \ \ \ \ \ \ \ \ k\in \mathcal{K},\textbf{A} \in \mathcal{F}^K,f\in \mathcal{F},\nonumber\\
%&\ \ \ b^I(f,\textbf{A})\geq R_{f,3}^k\textbf{1}(A_k=f)x_{f,3}^k,\ k \in \mathcal{K}, \ \textbf{A} \in \mathcal{F}^K,f\in \mathcal{F},\nonumber\\
%&\ \ \ \ \  a^O(f,\textbf{A}) \geq \frac{1}{\log\left(1+\frac{Ph_k^2}{\sigma^2}\right)}\textbf{1}(A_k=f)x_{f,4}^k,\  \nonumber\\
%&\ \ \ \ \ \ \ \ \ \ \ \ \ \ \ \ \ \ \ \ \ \ \ \ \ \ \ \ \ \ k\in \mathcal{K},\textbf{A} \in \mathcal{F}^K,f\in \mathcal{F},\nonumber\\
%&\ \ \ b^O(f,\textbf{A})\geq R_{f,4}^k\textbf{1}(A_k=f)x_{f,4}^k,\ k \in \mathcal{K}, \ \textbf{A} \in \mathcal{F}^K,f\in \mathcal{F}, \nonumber
\end{align}
\textcolor{black}{with} the penalty parameter $\rho>0$.
\end{Prob}

%Note that the objective function of Problem~\ref{penalized} is a difference of two convex functions, and the constraints of Problem~\ref{penalized} are linear. Thus, Problem~\ref{penalized} is a DC problem.

Based on Theorem~5 and Theorem~8 in \cite{exactpenalty}, we show the equivalence between Problem~\ref{dcproblem} and Problem~\ref{penalized} in the following lemma.
%. By Theorem~1 in \cite{exactpenalty}, we obtain the following lemma.

\begin{lemma}[Exact Penalty]\label{exact} There exists $\rho_0>0$ such that when $\rho \geq  \rho_0$, Problem~\ref{penalized} and Problem~\ref{dcproblem} have the same optimal solution.
\end{lemma}
%\begin{proof}Please see Appendix~B.
%\end{proof}
Lemma~\ref{exact} illustrates that Problem~\ref{penalized} is equivalent to Problem~\ref{dcproblem} if the penalty parameter $\rho$ is sufficiently large. Thus, in the sequel, we solve Problem~\ref{penalized} instead of Problem~\ref{dcproblem} by using CCCP to obtain the stationary point \cite{exactpenalty}. In general, CCCP involves iteratively solving a sequence of convex subproblems, each of which is obtained via linearizing the concave-term of the objective function of Problem~\ref{penalized}, i.e., replacing the concave parts with their first-order Taylor expansions. Specifically, in the $t$-th iteration, we need to solve
\begin{Prob}[CCCP Subproblem in the $t$-th Iteration]\label{subproblem}
\begin{align}
& \min_{\{\textbf{a}^I,\ \textbf{b}^I,\ \textbf{a}^O,\ \textbf{x}\}} \ \frac{1}{N} \sum_{n=1}^N\sum_{f=1}^F \Bigg[\frac{\left(a_{f,n}^I+b_{f,n}^I\right)^2}{4}+\frac{O_f}{\tau}a_{f,n}^O\nonumber\\
&\ \ \ \ \ \ \ \ \ \ \ \ \ \ \ \ \ \ -\frac{a_{f,n}^I(t)-b_{f,n}^I(t)}{2}\left(a_{f,n}^I-b_{f,n}^I\right)\Bigg]\nonumber\\
&\ \ \ \ \ \ \ \ \ \ \ \ \ \ \ \ \ \ - \rho\sum_{k=1}^K\sum_{f=1}^F\sum_{j=1}^4\left(2x_{f,j}^k(t)-1\right)x_{f,j}^k\nonumber \\
&\ \ \  \ \ s.t. \ \ \ \ \ \ \ \  (\ref{schedule}), (\ref{cachesize2}), (\ref{energy2}), (\ref{ai}), (\ref{bi}), (\ref{ao}), (\ref{aa}), \nonumber
%&\ \ \ \ \ a^I(f,\textbf{A}) \geq \frac{1}{\log\left(1+\frac{Ph_k^2}{\sigma^2}\right)}\textbf{1}(A_k=f)x_{f,3}^k,\nonumber\\
%&\ \ \ \ \ \ \ \ \ \ \ \ \ \ \ \ \ \ \ \ \ \ \ \ \ \ \ \ \ \ k\in \mathcal{K},\textbf{A} \in \mathcal{F}^K,f\in \mathcal{F},\nonumber\\
%&\ \ \ b^I(f,\textbf{A})\geq R_{f,3}^k\textbf{1}(A_k=f)x_{f,3}^k,\ k \in \mathcal{K}, \ \textbf{A} \in \mathcal{F}^K,f\in \mathcal{F},\nonumber\\
%&\ \ \ \ \  a^O(f,\textbf{A}) \geq \frac{1}{\log\left(1+\frac{Ph_k^2}{\sigma^2}\right)}\textbf{1}(A_k=f)x_{f,4}^k,\  \nonumber\\
%&\ \ \ \ \ \ \ \ \ \ \ \ \ \ \ \ \ \ \ \ \ \ \ \ \ \ \ \ \ \ k\in \mathcal{K},\textbf{A} \in \mathcal{F}^K,f\in \mathcal{F},\nonumber\\
%&\ \ \ b^O(f,\textbf{A})\geq R_{f,4}^k\textbf{1}(A_k=f)x_{f,4}^k,\ k \in \mathcal{K}, \ \textbf{A} \in \mathcal{F}^K,f\in \mathcal{F}, \nonumber
\end{align}
where $\textbf{a}^I \triangleq (a_{f,n}^I)_{f\in \mathcal{F},n\in \mathcal{N}}$,  $\textbf{b}^I \triangleq (b_{f,n}^I)_{f\in \mathcal{F},n\in \mathcal{N}}$, $\textbf{a}^O \triangleq (a_{f,n}^O)_{f\in \mathcal{F},n\in \mathcal{N}}$, and $\Big\{(a_{f,n}^I(t))_{f\in \mathcal{F},n\in \mathcal{N}},(b_{f,n}^I(t))_{f\in \mathcal{F},n\in \mathcal{N}},\\ (x_{f,j}^k(t))_{f\in \mathcal{F}, j\in \{1,2,3,4\}, k\in \mathcal{K}}\Big\}$ are the optimal solution obtained from the last iteration, and the penalty parameter $\rho>0$.
\end{Prob}

Note that Problem~\ref{subproblem} is a convex problem and can be solved via a general-purpose solver based on interior-point methods.  However, it may suffer from high computation complexity due to the large sample size $N$. In the following, we exploit the specific structure of Problem~\ref{subproblem} and find its optimal solution using an ADMM algorithm \cite{erkai}.
\subsubsection{ADMM algorithm for each CCCP subproblem}
First, we introduce a set of consensus constraints $x_{f,j}^{k,n} = x_{f,j}^k,\ f\in \mathcal{F},\ j\in \{1,2,3,4\},\ k\in \mathcal{K},\ n\in \mathcal{N} $, and reformulate Problem~\ref{subproblem} as
\begin{Prob}[Equivalent CCCP Subproblem in the $t$-th Iteration]\label{subproblem1}
\begin{align}
& \min_{\{\textbf{a}^I,\ \textbf{b}^I,\ \textbf{a}^O,\ \textbf{x}\}} \ \frac{1}{N} \sum_{n=1}^N\sum_{f=1}^F \Bigg[\frac{\left(a_{f,n}^I+b_{f,n}^I\right)^2}{4}+\frac{O_f}{\tau}a_{f,n}^O\nonumber\\
&\ \ \ \ \ \ \ \ \ \ \ \ \ \ \ \ \  -\frac{a_{f,n}^I(t)-b_{f,n}^I(t)}{2}\left(a_{f,n}^I-b_{f,n}^I\right)\Bigg]\nonumber\\
&\ \ \ \ \ \ \ \ \ \ \ \ \ \ \ \  - \rho\frac{1}{N}\sum_{n=1}^N\sum_{k=1}^K\sum_{f=1}^F\sum_{j=1}^4\left(2x_{f,j}^k(t)-1\right)x_{f,j}^{k,n}\label{objective} \\
&\ \ \  \ \ s.t. \ \ \ \ \ \ \ \  (\ref{schedule}), (\ref{cachesize2}), (\ref{energy2}), (\ref{aa}), \nonumber\\
&\ \ \ \ \ \ \ \ \ \ \ a^I_{f,n} \geq \frac{1}{\log\left(1+\frac{Ph_k^2}{\sigma^2}\right)}\textbf{1}(A_{n,k}=f)x_{f,3}^{k,n},\nonumber\\
&\hspace{40mm}k\in \mathcal{K},\ f\in \mathcal{F}, \ n \in \mathcal{N},\label{ai1}\\
&\ \ \ \ \ \ \ \ \ \ \ b^I_{f,n} \geq R_{f,3}^k \textbf{1}(A_{n,k}=f) x_{f,3}^{k,n},\nonumber\\
&\hspace{40mm} k\in \mathcal{K},\  f \in \mathcal{F}, \ n \in \mathcal{N},\label{bi1}\\
&\ \ \ \ \ \ \ \ \ \ \ a^O_{f,n} \geq \frac{1}{\log\left(1+\frac{Ph_k^2}{\sigma^2}\right)}\textbf{1}(A_{n,k}=f)x_{f,4}^{k,n},\nonumber\\
&\hspace{40mm} k \in \mathcal{K},\ f\in \mathcal{F},\ n \in \mathcal{N}, \label{ao1}\\
&\ \ \ \ \ \ \ \ \ \ \ x_{f,j}^{k,n} = x_{f,j}^k,\nonumber\\
&\ \ \ \ \ \ \ \ \ \ \ \ \ \ \  f\in \mathcal{F},\ j\in \{1,2,3,4\},\ k\in \mathcal{K},\ n\in \mathcal{N}, \label{consensus}
\end{align}
where $\textbf{x} = (x_{f,j}^k)_{f\in \mathcal{F}, j\in \{1,2,3,4\},k\in \mathcal{K}}$ in constraints (\ref{ai}), (\ref{bi}) and (\ref{ao}) is replaced with $\{\textbf{x}^n\}_{n\in \mathcal{N}} \triangleq (x_{f,j}^{k,n})_{f\in \mathcal{F}, j\in \{1,2,3,4\},k\in \mathcal{K},n\in \mathcal{N}}$ in constraints (\ref{ai1}), (\ref{bi1}) and (\ref{ao1}).
\end{Prob}

Then, drop the constant $1/N$ in (\ref{objective}) and we obtain the partial augmented Lagrangian of Problem~\ref{subproblem1} via moving the consensus constraint (\ref{consensus}) to the objective function of Problem~\ref{subproblem1} as follows:
\begin{align}
&\mathcal{L}_{\gamma} \left(\textbf{a}^I,\textbf{b}^I,\textbf{a}^O, (\textbf{x}^n)_{n\in \mathcal{N}}, \textbf{x}; (\mathbf{\lambda}^n)_{n\in \mathcal{N}}\right)\nonumber\\
&\triangleq  \sum_{n=1}^N\sum_{f=1}^F \Bigg[\frac{\left(a_{f,n}^I+b_{f,n}^I\right)^2}{4}+\frac{O_f}{\tau}a_{f,n}^O\nonumber\\
&\ \ \ \ \ \ \ \ \ \ \ \ \  -\frac{a_{f,n}^I(t)-b_{f,n}^I(t)}{2}\left(a_{f,n}^I-b_{f,n}^I\right)\Bigg]\nonumber\\
& - \rho \sum_{n=1}^N\sum_{k=1}^K\sum_{f=1}^F\sum_{j=1}^4\left(2x_{f,j}^k(t)-1\right)x_{f,j}^{k,n}\nonumber \\
&+ \sum_{n=1}^N \sum_{k=1}^K \sum_{f=1}^F \sum_{j=1}^4 \left [ \lambda_{f,j}^{k,n}\left(x_{f,j}^{k,n}-x_{f,j}^k\right)+\frac{\gamma}{2}\left(x_{f,j}^{k,n}-x_{f,j}^k\right)^2\right],\label{augmented}
\end{align}
where $\mathbf{\lambda}^n \triangleq (\lambda_{f,j}^{k,n})_{f\in \mathcal{F}, j\in \{1,2,3,4\}, k \in \mathcal{K}}$, $\lambda_{f,j}^{k,n}$ is the Lagrangian multiplier corresponding to the constraint $x_{f,j}^{k,n} = x_{f,j}^k$ and $\gamma>0$ is the penalty parameter.

In general, ADMM involves iteratively updating the primal varaibles via minimizing the augmented Lagrangian (\ref{augmented}),  and then updating the Lagrangian multiplier. In particular, ADMM updates the variables at iteration $q+1$ according to the following three steps \cite{liuyafeng,erkai,admm}:
\begin{itemize}
\item \textit{$\left\{\textbf{a}^I, \textbf{b}^I, \textbf{a}^O, (\textbf{x}^n)_{n\in \mathcal{N}}\right\}$  Update}. Given $\left\{\textbf{x}, (\mathbf{\lambda}^n)_{n\in \mathcal{N}}\right\}^q$ obtained from iteration $q$, update $\left\{\textbf{a}^I, \textbf{b}^I, \textbf{a}^O, (\textbf{x}^n)_{n\in \mathcal{N}}\right\}$ for iteration $q+1$ as the solution to the following problem:
\begin{align}
&\min_{\left\{\textbf{a}^I, \textbf{b}^I, \textbf{a}^O, (\textbf{x}^n)_{n\in \mathcal{N}}\right\}}  \mathcal{L}_{\gamma} \Big(\textbf{a}^I,\textbf{b}^I,\textbf{a}^O, (\textbf{x}^n)_{n\in \mathcal{N}}, \{\textbf{x}\}^q;\nonumber\\
& \ \ \ \ \ \ \ \ \ \ \ \ \ \ \ \ \ \ \ \ \ \ \ \ \ \ \ \ \ \ \ \{(\mathbf{\lambda}^n)_{n\in \mathcal{N}}\}^q \Big)\nonumber\\
&\ \ \ \ \ \ \ \ \ s.t. \ \ \ \ \ \ \ \ \ \ \ (\ref{ai1}),\ (\ref{bi1}),\ (\ref{ao1}).\nonumber
\end{align}
\item \textit{$\textbf{x}$ Update}. Given $\left\{\textbf{a}^I, \textbf{b}^I, \textbf{a}^O, (\textbf{x}^n)_{n\in \mathcal{N}}\right\}^{q+1}$ obtained from iteration $q+1$ \textcolor{black}{and $\left\{ (\mathbf{\lambda}^n)_{n\in \mathcal{N}}\right\}^q$ obtained from iteration $q$}, update $\textbf{x}$ for iteration $q+1$ as the solution to the following problem:
\begin{align}
&\min_{\textbf{x}} \  \mathcal{L}_{\gamma} \Big(\left\{\textbf{a}^I,\textbf{b}^I,\textbf{a}^O, (\textbf{x}^n)_{n\in \mathcal{N}}\right\}^{q+1}, \textbf{x}; \{(\mathbf{\lambda}^n)_{n\in \mathcal{N}}\}^q \Big)\nonumber\\
%& \ \ \ \ \ \ \ \ \ \ \ \ \ \ \ \ \ \ \ \ \ \ \ \ \ \ \ \ \ \ \ \{(\mathbf{\lambda}^n)_{n\in \mathcal{N}}\}^q \Big)\nonumber\\
&\ s.t. \ \ \ \ \ \ \ \ \ \ \ (\ref{schedule}), (\ref{cachesize2}), (\ref{energy2}), (\ref{aa}).\nonumber
\end{align}
\item \textit{$(\mathbf{\lambda}^n)_{n\in \mathcal{N}}$ Update}. Given $\left\{\textbf{a}^I, \textbf{b}^I, \textbf{a}^O, (\textbf{x}^n)_{n\in \mathcal{N}},\textbf{x}\right\}^{q+1}$ obtained from iteration $q+1$, update $(\mathbf{\lambda}^n)_{n\in \mathcal{N}}$ for iteration $q+1$ according to:
\begin{align}
&\left\{\lambda_{f, j}^{k,n}\right\}^{q+1} = \left\{\lambda_{f, j}^{k,n}\right\}^{q} + \gamma\left(\left\{x_{f,j}^{k,n}\right\}^{q+1} - \left\{x_{f,j}^{k}\right\}^{q+1}\right),\nonumber\\
&\ \ \ \ \ \ f\in \mathcal{F},\ j\in \{1,2,3,4\},\ k\in \mathcal{K},\ n\in \mathcal{N}.
\end{align}
\end{itemize}

For the update of $\left\{\textbf{a}^I, \textbf{b}^I, \textbf{a}^O, (\textbf{x}^n)_{n\in \mathcal{N}}\right\}$, the optimization problem is decoupled among the $N$ request realizations and the $F$ files, and into $NF$ subproblems. For each $n \in \mathcal{N}$ and $f\in \mathcal{F}$, we solve the following subproblem:
\begin{align}
&\min_{\left\{a_{f,n}^I,b_{f,n}^I,a_{f,n}^O, \textbf{x}_f^n\right\}}  \frac{\left(a_{f,n}^I+b_{f,n}^I\right)^2}{4}+\frac{O_f}{\tau}a_{f,n}^O\nonumber\\
&\ \ \ \ \ \ \ \  -\frac{a_{f,n}^I(t)-b_{f,n}^I(t)}{2}\left(a_{f,n}^I-b_{f,n}^I\right)\nonumber\\
&\ \ \ \ \ \ \ \ - \rho \sum_{k=1}^K\sum_{j=1}^4\left(2x_{f,j}^k(t)-1\right)x_{f,j}^{k,n}\nonumber \\
&\ \ \ \ \ \ \ \ +  \sum_{k=1}^K \sum_{j=1}^4 \left [\left\{ \lambda_{f,j}^{k,n}\right\}^q x_{f,j}^{k,n}+\frac{\gamma}{2}\left(x_{f,j}^{k,n}-\left\{x_{f,j}^k\right\}^q\right)^2\right]\nonumber\\
&\ \ \ \ \ \ \ s.t. \ a^I_{f,n} \geq \frac{1}{\log\left(1+\frac{Ph_k^2}{\sigma^2}\right)}\textbf{1}(A_{n,k}=f)x_{f,3}^{k,n},\ k\in \mathcal{K},\nonumber\\
%&\hspace{40mm}k\in \mathcal{K},\label{ai1}\\
&\ \ \ \ \ \ \ \ \ \ \ \ b^I_{f,n} \geq R_{f,3}^k \textbf{1}(A_{n,k}=f) x_{f,3}^{k,n},\ k\in \mathcal{K},\nonumber\\
%&\hspace{40mm} k\in \mathcal{K},\label{bi1}\\
& \ \ \ \ \ \ \ \ \ \ \ \ a^O_{f,n} \geq \frac{1}{\log\left(1+\frac{Ph_k^2}{\sigma^2}\right)}\textbf{1}(A_{n,k}=f)x_{f,4}^{k,n},\ k\in \mathcal{K}. \nonumber
%&\hspace{40mm} k \in \mathcal{K}, \label{ao1}
\end{align}

Based on KKT conditions, the closed-form expression for the optimal solution $\left\{\textbf{a}^I, \textbf{b}^I, \textbf{a}^O, (\textbf{x}^n)_{n\in \mathcal{N}}\right\}^{q+1} $ is given by
\begin{align}
&a^I_{f,n} = \max_{k\in \mathcal{K}} \frac{1}{\log\left(1+\frac{Ph_k^2}{\sigma^2}\right)}\textbf{1}(A_{n,k}=f)\left\{x_{f,3}^{k,n}\right\}^{q+1},\ \nonumber\\
%&\hspace{40mm}k\in \mathcal{K},\label{ai1}\\
& b^I_{f,n} = \max_{k\in \mathcal{K}} R_{f,3}^k \textbf{1}(A_{n,k}=f) \left\{x_{f,3}^{k,n}\right\}^{q+1},\nonumber\\
%&\hspace{40mm} k\in \mathcal{K},\label{bi1}\\
&a^O_{f,n} = \max_{k\in \mathcal{K}} \frac{1}{\log\left(1+\frac{Ph_k^2}{\sigma^2}\right)}\textbf{1}(A_{n,k}=f)\left\{x_{f,4}^{k,n}\right\}^{q+1}, \nonumber
\end{align}
\begin{align}
&\left\{x_{f,j}^{k,n}\right\}^{q+1} = \frac{\rho}{\gamma} \left(2x_{f,j}^k(t)-1\right)-\frac{1}{\gamma} \left\{\lambda_{f,j}^{k,n}\right\}^q + \left\{x_{f,j}^k\right\}^q,\nonumber\\
&\hspace{54mm} j\in \{1,2\},\ k\in \mathcal{K}, \label{12}\nonumber\\
&\left\{x_{f,3}^{k,n}\right\}^{q+1} = \frac{\rho}{\gamma} \left(2x_{f,3}^k(t)-1\right)-\frac{1}{\gamma} \left\{\lambda_{f,3}^{k,n}\right\}^q + \left\{x_{f,3}^k\right\}^q\nonumber\\
&\hspace{17mm} - \frac{u_{1,k}}{\gamma}R_{k,f}^3 \textbf{1}\left(A_{n,k}=f\right)\nonumber\\
&\hspace{17mm}-\frac{u_{2,k}}{\gamma}\frac{1}{\log\left(1+\frac{Ph_k^2}{\sigma^2}\right)}\textbf{1}\left(A_{n,k}=f\right),\ k \in \mathcal{K},\nonumber\\
&\left\{x_{f,4}^{k,n}\right\}^{q+1} = \frac{\rho}{\gamma} \left(2x_{f,4}^k(t)-1\right)-\frac{1}{\gamma} \left\{\lambda_{f,4}^{k,n}\right\}^q + \left\{x_{f,4}^k\right\}^q\nonumber\\
&\hspace{17mm}-\frac{u_{3,k}}{\gamma}\frac{1}{\log\left(1+\frac{Ph_k^2}{\sigma^2}\right)}\textbf{1}\left(A_{n,k}=f\right),\ k\in \mathcal{K},\nonumber
\end{align}
where $(u_{1,k})_{k\in \mathcal{K}}$, $(u_{2,k})_{k\in \mathcal{K}}$ and $(u_{3,k})_{k\in \mathcal{K}}$ are given by
\begin{align}
u_{1,k} =
&\left\{
\begin{array}{l}
0,\ \ \ \ \ \ \ \ k \neq k^*_1,\\
u_{1,k_1^*},\ \ \  k = k^*_1,
\end{array}
\right.
\end{align}
\begin{align}
u_{2,k} =
&\left\{
\begin{array}{l}
0,\ \ \ \ \ \ \ \ k \neq k^*_2,\\
u_{2,k_2^*},\ \ \ k = k^*_2,
\end{array}
\right.
\end{align}
\begin{align}
u_{3,k} =
&\left\{
\begin{array}{l}
0,\ \ \ \ \ \ \ \ k \neq k^*_3,\\
\frac{O_f}{\tau},\ \ \ \ \ \ k = k^*_3,
\end{array}
\right.
\end{align}
with $k_1^* \triangleq \arg \max_{k\in \mathcal{K}} R_{k,f}^3\textbf{1}\left(A_{n,k}=f\right)\left\{x_{f,3}^{k,n}\right\}^{q+1}$, $k_2^* \triangleq \arg \max_{k\in \mathcal{K}} \frac{1}{\log\left(1+\frac{Ph_k^2}{\sigma^2}\right)}\textbf{1}\left(A_{n,k}=f\right) \left\{x_{f,3}^{k,n}\right\}^{q+1}$, and $u_{1,k_1^*}$ and $u_{2,k_2^*}$ satisfying
\begin{align}
&\frac{R_{k_1^*,f}^3\textbf{1}\left(A_{n,k_1^*}=f\right)}{2}\left\{x_{f,3}^{k_1^*,n}\right\}^{q+1} \nonumber\\
&+ \frac{\textbf{1}\left(A_{n,k_2^*}=f\right)}{2\log\left(1+\frac{Ph_{k_2^*}^2}{\sigma^2}\right)}\left\{x_{f,3}^{k_2^*,n}\right\}^{q+1}-\frac{a_{f,n}^I(t)-b_{f,n}^I(t)}{2}\nonumber\\
& = u_{1,k_1^*},
\end{align}
and $k_3^* \triangleq \arg \max_{k\in \mathcal{K}} \frac{1}{\log\left(1+\frac{Ph_k^2}{\sigma^2}\right)}\textbf{1}\left(A_{n,k}=f\right) \left\{x_{f,4}^{k,n}\right\}^{q+1}$.

Note that $k_1^*$, $k_2^*$ and $k_3^*$ can be determined via brute-force search on the complexity order $\mathcal{O}(K^3)$, and then $\left\{\textbf{a}^I, \textbf{b}^I, \textbf{a}^O, (\textbf{x}^n)_{n\in \mathcal{N}}\right\}^{q+1} $ is obtained directly.

For the update of $\textbf{x}$, the optimization problem is decomposed among the request realizations and the $K$ users, and into $NK$ subproblems. For each $n\in \mathcal{N}$ and $k\in \mathcal{K}$, we solve the following subproblem:
\begin{align}
&\min_{(x_{f,j}^k)_{f\in \mathcal{F}, j\in \{1,2,3,4\}}}\ \sum_{f=1}^F \sum_{j=1}^4 \Bigg[ \lambda_{f,j}^{k,n}\left(x_{f,j}^{k,n}-\left\{x_{f,j}^k\right\}^q\right)\nonumber\\
&\hspace{30mm}+\frac{\gamma}{2}\left(x_{f,j}^{k,n}-\left\{x_{f,j}^k\right\}^q\right)^2\Bigg] \nonumber \\
&\ \ \ \ \ \ \ \ \ s.t.\ \ \ \ \ \ \ \ \ \ \ (\ref{schedule}), (\ref{cachesize2}), (\ref{energy2}), (\ref{aa}). \nonumber
\end{align}

Based on KKT conditions, we can obtain the optimal solution $(x_{f,j}^k)_{f\in \mathcal{F}, j\in \{1,2,3,4\}} $ in the closed-form expression as following
\begin{align}
&\left\{x_{f,1}^k\right\}^{q+1} = \min \left\{\max \left\{S_{k,1}^k - \frac{\lambda_1}{2}O_f - \frac{\mu_f}{2}, 0\right\}, 1\right\}, \nonumber\\
&\left\{x_{f,2}^k\right\}^{q+1} = \min \Bigg\{\max \Big\{S_{k,2}^k - \frac{\lambda_1}{2}I_f -\frac{\lambda_2}{2}P_f\mu I_f w_f f_k^2\nonumber\\
&\hspace{40mm}- \frac{\mu_f}{2}, 0\Big\}, 1 \Bigg\}, \nonumber\\
&\left\{x_{f,3}^k\right\}^{q+1} = \min \Bigg\{\max \Big\{S_{k,3}^k -\frac{\lambda_2}{2}P_f\mu I_f w_f f_k^2- \frac{\mu_f}{2}, 0\Big\}, \nonumber\\
&\hspace{45mm}1 \Bigg\}, \nonumber\\
&\left\{x_{f,4}^k\right\}^{q+1} = \min \left\{\max \left\{S_{k,4}^k - \frac{\mu_f}{2}, 0\right\}, 1\right\}, \nonumber
\end{align}
where $\lambda_1$, $\lambda_2$ and $(\mu_f)_{f\in \mathcal{F}}$ are the optimal Lagrangian multipliers satisfying (\ref{schedule}), $\sum_{f=1}^F O_fx_{f,1}^k+I_fx_{f,2}^k = C_k$ and $\sum_{f=1}^F P_f\mu I_fw_ff_k^2(x_{f,2}^k+x_{f,3}^k) = \bar{E}_k$.

Based on Section 3.2 in \cite{admm} and Proposition~15 in \cite{admm2}, the above mentioned ADMM is guaranteed to converge to the optimal solution to Problem~\ref{subproblem1}. Thus, our proposed CCCP-ADMM algorithm, as illustrated in Algorithm~\ref{dca}, converges to a stationary point of our original problem. \textcolor{black}{Based on CCCP-ADMM, we can see that the computation complexity is reduced from $\mathcal{O}\left(2^{3K}\right)$ to $\mathcal{O}\left(N\max\{F,K\}\right)$.
}

\begin{algorithm}[t]
 %\floatname{algorithm}{Online Decentralized Algorithm (ODA)}
 %\setcounter{algorithm}{0}
 \caption{CCCP-ADMM algorithm}
\label{dca}
\small{\begin{algorithmic}[1]
%\STATe \textbf{Input}. $K,F,\{S_f, \forall f\},\{C_k, \forall k\}$
%\STATe \textbf{Output}. Cache placement \textbf{C}
\STATE \textbf{Initialization}. Find an initial feasible point $\textbf{x}^{(0)}$  of Problem~\ref{subproblem1} and set $t=0$.
\STATE \textbf{Repeat}
\STATE  Obtain $\textbf{x}^{(t+1)}$ the optimal solution to the $t+1$-th subproblem via ADMM algorithm.
\STATE Set $t\leftarrow t+1$.
\STATE \textbf{until} $\Big[G\left(t-1\right)-H(t-1;t-2)\Big]- \left[G \left(t\right) - H \left(t;t-1\right)\right] \leq \delta$,
where $G\left(t\right) \triangleq  \frac{1}{N} \sum_{n=1}^N\sum_{f=1}^F \frac{\left(a_{f,n}^I(t)+b_{f,n}^I(t)\right)^2}{4}+\frac{O_f}{\tau}a_{f,n}^O(t)$ and $H(t;t-1) \triangleq \frac{1}{N}\sum_{n=1}^N\sum_{f=1}^F\frac{a_{f,n}^I(t-1)-b_{f,n}^I(t-1)}{2}\left(a_{f,n}^I(t)-b_{f,n}^I(t)\right)+\rho\frac{1}{N}\sum_{n=1}^N\sum_{k=1}^K\sum_{f=1}^F\sum_{j=1}^4\left(2x_{f,j}^k(t-1)-1\right)x_{f,j}^{k,n}(t)$.% -\frac{a_{f,n}^I(t)-b_{f,n}^I(t)}{2}\left(a_{f,n}^I-b_{f,n}^I\right)\bigg]$ and $a$
\end{algorithmic}}
\end{algorithm}

\section{Bandwidth Gain Analysis}
In this section, we analyze the symmetric scenario to gain more design insights, i.e., for all $f\in \mathcal{F}, \ k\in \mathcal{K}$, $(I_f,w_f,O_f) = (I,w,O)$, $P_{k,f} = \frac{1}{F}$, $f_k=f_1$, $C_k=C$, $\bar{E}_k = \bar{E}$ and $h_k=h$. Accordingly, we have $R_{f,3}^k = R_3$ and $R_{f,4}^k = R_4$, where $R_3 \triangleq \frac{I}{\tau-\frac{Iw}{f_1}}$ and $R_4 \triangleq \frac{O}{\tau}$, for all $f\in \mathcal{F}$ and $k\in \mathcal{K}$. \textcolor{black}{For notation convenience, we define $\beta_{c} \triangleq \frac{C}{FO} \leq 1$, which represents the normalized cache size at each device with respect to the total output data size of all the tasks; denote $\beta_{e} \triangleq \frac{\bar{E}}{\mu Iwf_1^2}\leq 1$, which represents the normalized average energy at each device with respect to the total average energy of all the tasks.}%maximum number of tasks that can be computed locally at each mobile device subject to the given energy constraint $\bar{E}$.  %interest of design, we assume that $C \leq OF$ and $\frac{F\bar{E}}{\mu Iwf_1^2} \leq F$.
\subsection{Optimal Policy}
First, by analyzing the structure of the problem, we obtain the optimal policy in the symmetric scenario, given as below.
\begin{lemma}[Optimal policy in symmetric scenario]\label{symmetric}
For all $k\in \mathcal{K}$,
\begin{equation}\label{optB1}
x_{f,1}^{k,*} =
\begin{cases}
1,& f = 1,\cdots,n_1, \\
0,& \text{otherwise},
\end{cases}
\end{equation}
where $n_1 \triangleq F\max \left\{\beta_c-\min\left\{\beta_c,\frac{1}{\alpha}\beta_e\right\}\textbf{1}(\alpha>1),0\right\}$,%$n_1 \triangleq \max\left\{\frac{C-\min\left\{C,\frac{F\bar{E}}{\mu w f_1^2}\right\}\textbf{1}(\alpha>1)}{O},0\right\}$,
\begin{equation}\label{optB2}
x_{f,2}^{k,*} =
\begin{cases}
1,& f=n_1+1,\cdots,n_1+n_2,\\
0,&\text{otherwise},
\end{cases}
\end{equation}
where $n_2 \triangleq F\min \left\{\alpha \beta_c,\beta_e\right\}\textbf{1}(\alpha>1)$,
\begin{equation}\label{optC}
x_{f,3}^{k,*} =
\begin{cases}
1,& f=n_1+n_2+1,\cdots,n_1+n_2+n_3,\\
0,&\text{otherwise},
\end{cases}
\end{equation}
where $n_3\! \triangleq\! F\left(\beta_e\!-\! \min \left\{\alpha \beta_c,\beta_e\right\}\right)\textbf{1}\left(\alpha\!>1, f_1>\!\frac{Iw}{(1\!-\!\frac{1}{\alpha})\tau}\!\right)$,% if $f_1 \leq \frac{Iw}{(1-\frac{1}{\alpha})\tau}$, and $n_3 \triangleq \frac{F\bar{E}}{\mu Iwf_1^2}\textbf{1}(\alpha>1)$ otherwise,
\begin{equation}
x_{f,4}^{k,*} =
\begin{cases}
1,&f = n_1+n_2+n_3+1,\cdots,F,\\
0,& \text{otherwise}.
\end{cases}
\end{equation}
\end{lemma}
\begin{proof}
Please see Appendix B.
\end{proof}
From Lemma~\ref{symmetric}, note that when $\alpha \leq 1$, $x_{f,2}^{k,*}=x_{f,3}^{k,*}=0$ for all $k\in \mathcal{K}$ and $f\in \mathcal{F}$, meaning that joint local input caching and computing does not bring any bandwidth gain, and  the caching resources at all the mobile devices are utilized merely for output caching. \textcolor{black}{On the other hand, when $\alpha > 1$, from (\ref{optB1}) and (\ref{optB2}), we can see that caching at each mobile device is exploited for input caching first and then output caching if there still remains underutilized caching. From (\ref{optB2}) and (\ref{optC}), we can see that computing at each mobile device is exploited from local computing with local caching first and then local computing only if there still remains underutilized computing resource and also local computing frequency is large enough.}
\subsection{Bandwidth Gain from Local Caching and Computing}
Next, we analytically quantify the gain on the bandwidth requirement that caching and computing resources at the mobile devices can bring over MEC computing, i.e., the outputs of all the tasks are transmitted from the MEC server. Denote with $B_{MEC}^*$ the minimum bandwidth requirement via MEC computing. Based on Lemma~\ref{symmetric}, we obtain the following \textcolor{black}{theorem}.
\begin{theorem}[Bandwidth Gain from Local Caching and Computing]\label{MECgain}
When $\alpha \leq 1$, we have
\begin{equation}
\frac{B^*}{B_{MEC}^*} = 1-\beta_c,
%\frac{B_{MEC}^*}{B^*} = \frac{F}{F-\frac{C}{O}},
\end{equation}
%\begin{equation}
%\frac{B^*}{B_{MEC}^*} = 1-\alpha_0,
%\end{equation}
which decreases with $C$ but is independent of $f_1$.

When $\alpha > 1$ and $f_1 \geq \sqrt{\frac{F\bar{E}}{\mu wC}}$, we have
\begin{equation}
\frac{B^*}{B_{MEC}^*} = 1-\beta_c-\left(1-\frac{1}{\alpha}\right)\beta_e,%\frac{F}{F-\frac{C}{O}-(\alpha -1) \frac{F\bar{E}}{\mu Owf_1^2}},
%\frac{B_{MEC}^*}{B^*} = \frac{F}{F-\frac{C}{O}-(\alpha -1) \frac{F\bar{E}}{\mu Owf_1^2}},
\end{equation}
which decreases with $C$ and increases with $f_1$.

When $\alpha > 1$ and $\frac{Iw}{(1-\frac{1}{\alpha})\tau} < f_1 < \sqrt{\frac{F\bar{E}}{\mu wC}}$, we have
\begin{equation}\label{gain2}
\frac{B^*}{B_{MEC}^*}  = 1 - \alpha \beta_c - \left(1-\frac{\tau}{\alpha(\tau-\frac{Iw}{f_1})} \right)\left(\beta_e-\alpha \beta_c\right),
%\frac{B_{MEC}^*}{B^*} = \frac{F}{F-\frac{F\bar{E}}{\mu Iwf_1^2} + \frac{\tau}{\alpha(\tau-\frac{Iw}{f_1})} \left(\frac{F\bar{E}}{\mu Iwf_1^2}-\frac{C}{I}\right)},
\end{equation}
which decreases with $C$ and first decreases and then increases with $f_1$.

When $\alpha > 1$ and  $f_1 \leq \min \left\{\frac{Iw}{(1-\frac{1}{\alpha})\tau}, \sqrt{\frac{F\bar{E}}{\mu wC}}\right\}$, we also have
\begin{equation}\label{gain3}
\frac{B^*}{B_{MEC}^*} = 1- \alpha \beta_c,%\frac{F}{F-\frac{C}{I}},
\end{equation}
which decreases with $C$ and is independent of $f_1$.
\end{theorem}
\begin{proof}
Please see Appendix~C.
\end{proof}

\begin{remark}[\textcolor{black}{$\alpha \leq 1$}]
We can see from Theorem~\ref{MECgain} that when the size of output data is smaller than that of input data ($\alpha \leq 1$) in the symmetric scenario, the bandwidth benefits only from the local caching \textcolor{black}{and thus there is no need for local computing}.
% For $\alpha>1$, it benefits from both local caching and computing capabilities of the mobile device when the mobile device has powerful computing ability, i.e., $f_1 > \min \left\{\frac{Iw}{(1-\frac{1}{\alpha})\tau}, \sqrt{\frac{F\bar{E}}{\mu wC}}\right\}$.
\end{remark}

\begin{remark}[\textcolor{black}{$\alpha>1$}]
In Theorem~\ref{MECgain}, we also reveal the important fact that when $\alpha>1$, computing and caching do not affect the bandwidth gain independently, but \textcolor{black}{interact on} each other to get the bandwidth gain. For example, when the computing ability $f_1$ and the caching size $C$ of the mobile device satisfy the following relationship: $f_1 \geq \sqrt{\frac{F\bar{E}}{\mu wC}}$, the bandwidth gain is $1-\beta_c-\left(1-\frac{1}{\alpha}\right)\beta_e$. \textcolor{black}{Otherwise,} the bandwidth gain becomes (\ref{gain2}) or (\ref{gain3}).
\end{remark}

\subsection{Bandwidth Gain from Multicast}
Finally, we analytically quantify the bandwidth gain resulting from the multicast transmission over the unicast transmission, in which the MEC server transmits the requested datas to the mobile devices via $K$ independent unicast channels. The average bandwidth requirement for unicast transmission under $\textbf{x}$, denoted as $B_{unicast}(\textbf{x})$, is given by
\begin{equation}\label{unicast}
B_{unicast}(\textbf{x}) \!\triangleq\! \sum_{k=1}^K \sum_{f=1}^FP_{k,f}\!\sum_{j=1}^4\! R_{f,j}^k\frac{1}{\log(1+\frac{Ph_k^2}{\sigma^2})}x_{f,j}^k,
\end{equation}
and denote with $B^*_{unicast}$ the minimum required bandwidth for unicast transmission. Based on Lemma~\ref{symmetric}, we obtain the multicast gain $\frac{B^*}{B_{unicast}^*}$ as below.
\begin{theorem}[Bandwidth Gain from Multicast]\label{gain1} In the symmetric scenario, we have
\begin{equation}
\frac{B^*}{B_{unicast}^*} = \frac{F(1-(1-\frac{1}{F})^K)}{K},
\end{equation}
which increases with $\frac{F}{K}$.
\end{theorem}
\begin{proof}
Please see Appendix D.
\end{proof}

Theorem~\ref{gain1} shows that in the symmetric scenario, the multicast gain depends only on the number of users $K$ and that of tasks $F$, and is unrelated to the computing and caching capabilities of mobile device.
 \begin{figure}[t]
\begin{center}
 \includegraphics[width=8.5cm]{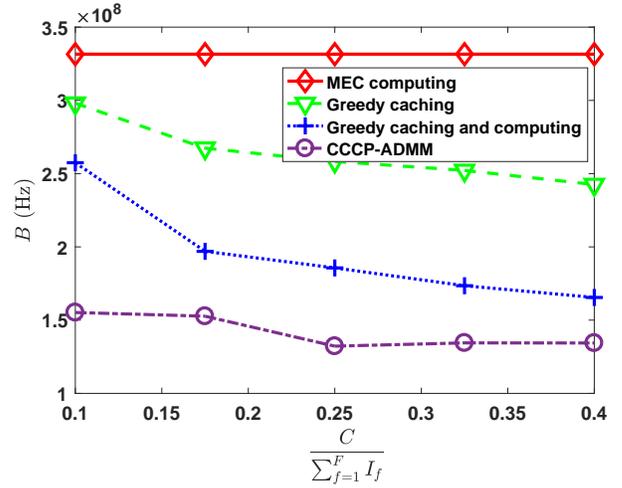}
\end{center}
 \caption{\small{Impact of $C$. $C_k=C,\ f_k=f_1$ and $\frac{1}{\log(1+\frac{Ph_k^2}{\sigma^2})} = 0.1*k$ for all $k\in \mathcal{K}$,  $F = 50$, $K = 4$, $I_f \in [10,15]$ M bits, $\alpha = 3$, $w = 10\ cycles/bit$, $\mu=10^{-27}$, $f_1= 1.1*10^{11}Hz$, $\bar{E} = 1.7*10^3J $, $P_{k,f} \propto \frac{1}{f^\gamma}$ with $\gamma= 1$, $\rho = 10^4$.} }
\label{heteC}
\end{figure}

\begin{figure}[t]
\begin{center}
 \includegraphics[width=8.2cm]{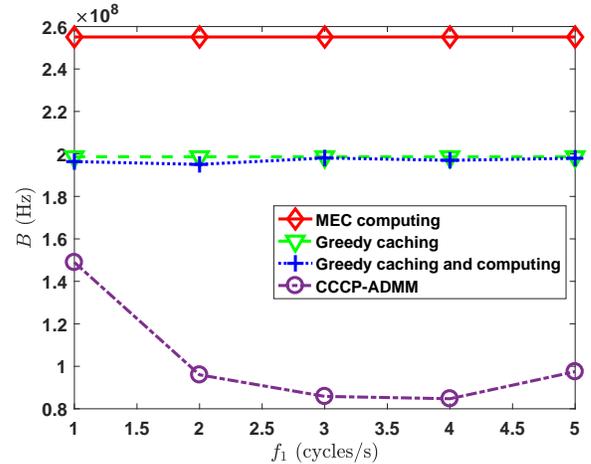}
\end{center}
 \caption{\small{Impact of $f_1$. $\frac{C}{\sum_{f=1}^F I_f} = 0.3$ and other parameters are the same as that in Fig.~\ref{heteC}.}
 }
\label{hetef1}
\end{figure}

\section{Numerical Results}
In this section, we present numerical results to evaluate the performance of the proposed CCCP-ADMM algorithm in terms of bandwidth saving. We compare it with the following three baselines:
\begin{itemize}
\item \textbf{MEC computing}: requests for all tasks are satisfied via Route~4, i.e., $x_{f,4}^k = 1$ for all $f\in \mathcal{F}$ and $k\in \mathcal{K}$;
\item \textbf{Greedy caching}: all the requests are satisfied via Route~1, i.e., for each user $k\in \mathcal{K}$, sort $\mathcal{F}$ according to $\frac{P_{k,f}R_{f,4}^k}{O_f}$ in descending order, denote with $\lfloor j\rfloor$ the index $f\in \mathcal{F}$ with the $j$-th maximal value of $\frac{P_{k,f}R_{f,4}^k}{O_f}$, and $s_c$ the split index satisfying $\sum_{j=1}^{s_c-1} O_{\lfloor j \rfloor} \leq C_k$ and $\sum_{j=1}^{s_c} O_{\lfloor j \rfloor} > C_k$. \textcolor{black}{Set $x_{\lfloor j\rfloor,1}^k = 1$, $x_{\lfloor j\rfloor,4}^k = 0$ for all $j\in \{1,\cdots,s_c\}$ and $x_{\lfloor j\rfloor,1}^k = 0$, $x_{\lfloor j\rfloor,4}^k = 1$, otherwise.} \textcolor{black}{$x_{f, i}^k = 0$ for all $i\in \{2,3\}$ and $f\in \mathcal{F}$}. Note that the complexity of this algorithm is $\mathcal{O}\left(KF\log(F)\right)$;%$c_{\lfloor j\rfloor}^O=1,c_{\lfloor j\rfloor}^I=0,d_{\lfloor j\rfloor}=0$ for all $j\in \{1,\cdots,s_c\}$ and $c_{\lfloor j\rfloor}^O=0,c_{\lfloor j\rfloor}^I=0,d_{\lfloor j\rfloor}=0$, otherwise;
\item \textbf{Greedy caching and computing}: for each user $k\in \mathcal{K}$, local input caching with local computing is determined via greedy algorithm, i.e., sort $\mathcal{F}$ according to $\frac{P_{k,f}R_{f,4}^k}{O_f+P_{k,f}\mu I_fw_f f_k^2}$ in descending order, denote with $\lfloor j\rfloor$ the index $f\in \mathcal{F}$ with the $j$-th maximal value of $\frac{P_{k,f}R_{f,4}^k}{O_f+P_{k,f}\mu I_fw_f f_k^2}$, and $s_c^1$ the split index satisfying $\sum_{j=1}^{s_c^1-1} I_{\lfloor j \rfloor} \leq C_k$ and $\sum_{j=1}^{s_c^1} I_{\lfloor j \rfloor} > C_k$ or $\sum_{j=1}^{s_c^1-1} P_{k,\lfloor j\rfloor} \mu I_{\lfloor j \rfloor}w_{\lfloor j\rfloor}f_k^2 \leq \bar{E}_k$ and $\sum_{j=1}^{s_c^1} P_{k,\lfloor j\rfloor}\mu I_{\lfloor j \rfloor}w_{\lfloor j\rfloor}f_k^2 > \bar{E}_k$. Set  $x_{\lfloor j \rfloor,2}^k = 1$ for all $j\in \{1,\cdots,s_c^1-1\}$, \textcolor{black}{and $x_{\lfloor j \rfloor,2}^k = 0$, otherwise}. Then, if there still exists underutilized cache size, i.e., $\sum_{j=1}^{s_c^1-1} I_{\lfloor j \rfloor} < C_k$, then outputs of the rest of tasks are cached at the mobile device via greedy algorithm. Otherwise, if $\sum_{j=1}^{s_c^1-1} P_{k,\lfloor j\rfloor}\mu I_{\lfloor j \rfloor}w_{\lfloor j\rfloor}f_k^2 < \bar{E}_k$, then local computing without caching is decided via greedy algorithm according to $\frac{P_{k,f}(R_{f,4}^k-R_{f,3}^k)}{P_{k,f}\mu I_f w_ff_k^2}$. Note that the complexity of this algorithm is $\mathcal{O}\left(KF\log(F)\right)$. 
\end{itemize}
\begin{figure}[t]
\begin{center}
 \includegraphics[width=8cm]{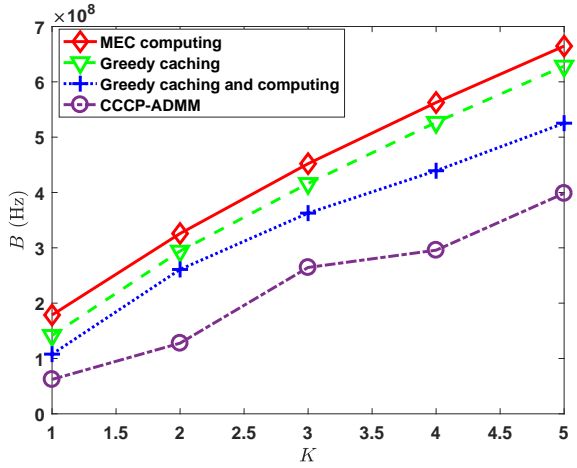}
\end{center}
 \caption{\small{Impact of $K$. $\frac{C}{\sum_{f=1}^F I_f} = 0.3$ and other parameters are the same as that in Fig.~\ref{heteC}.}
 }
\label{heteK}
\end{figure}

 \begin{figure}[t]
\begin{center}
 \includegraphics[width=7.6cm]{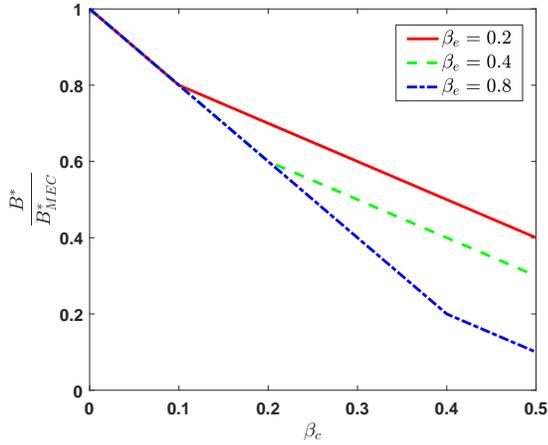}
\end{center}
 \caption{\small{Impact of $\beta_c$ on local caching and computing gain. $F = 50$, $K = 10$, $I = 15  M\ bits$, $w = 10\ cycles/bit$, $\alpha = 2$. }%, $f_1= 1.1*10^{11}Hz$, $\bar{E} = 1.7*10^3J $ \mu=10^{-27}
 }
\label{multigainC}
\end{figure}

%We illustrate the performance of CCCP in Fig.~12.

Fig.~\ref{heteC} and Fig.~\ref{hetef1} illustrate the impacts of the local cache size, i.e., $C$, and  computation frequency, i.e., $f_1$, on the average bandwidth cost, respectively. Fig.~\ref{heteK} illustrates the impact of the number of users, i.e., $K$, on the bandwidth requirement. We see that CCCP-ADMM exhibits great promises in saving  communication bandwidth compared with the baselines. For example,  in Fig.~2, compared with MEC computing, greedy caching and greedy caching and computing, CCCP-ADMM  brings significant transmission rate gain (e.g., 57.2\%,  42.3 \% vs. 25\% at $\frac{C}{\sum_{f=1}^F I_f} = 17.5\%$).
\begin{figure}[t]
\begin{center}
 \includegraphics[width=7.5cm]{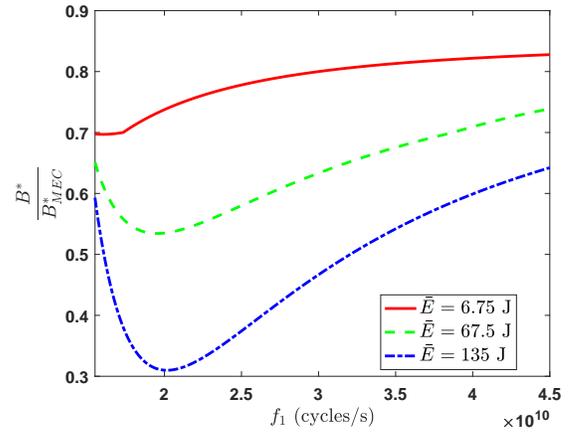}
\end{center}
 \caption{\small{Impact of $f_1$ on local caching and computing gain. $\beta_c = 0.3$, $ \mu=10^{-27}
$ and other parameters are the same as that in Fig.~\ref{multigainC}.}
 }
\label{multigainf1}
\end{figure}

\begin{figure}[t]
\begin{center}
 \includegraphics[width=7.4cm]{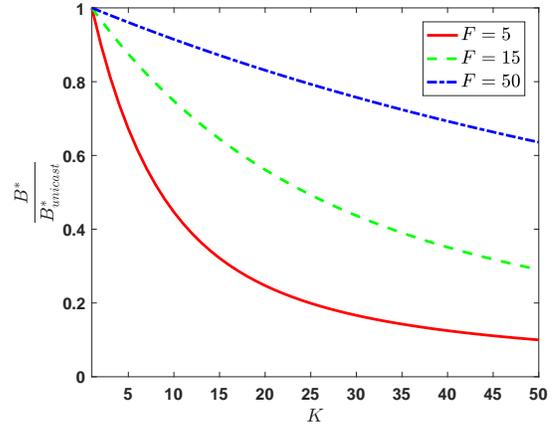}
\end{center}
 \caption{\small{Impact of $K$ on multicast gain.}
 }
\label{multigainK}
\end{figure}

In Fig.~\ref{multigainC}, we present the bandwidth gain versus the normalized caching size $\beta_c$ with different normalized average energy in the symmetric scenario. We can see that the local caching and computing gain increases with $\beta_c$ and the increasing rate depends on the relationship between $\beta_e$ and $\beta_c$. 

 From Fig.~\ref{multigainf1}, we can see that the local caching and computing gain decreases with $f_1$ when the average energy is limited, since increasing $f_1$ decreases the number of computation tasks that can be computed locally. However, the gain first increases and then decreases with $f_1 $ when the average energy is relatively large, since increasing $f_1$ decreases the transmission rate requirement.

From Fig.~\ref{multigainK}, we can see that the multicast gain increases with $K$. %In addition,  multicast transmission achieves bandwidth gain only when $F \leq K$, and not otherwise.
This is mainly because when the number of users $K$ increases, the probability that multiple users request the same task increases, and thus the multicast gain is growing.

\section{Conclusion}
In this paper, we investigate the impacts of the caching and computing resources at mobile devices on the transmission bandwidth, and optimize the joint caching and computing policy to minimize the average transmission bandwidth under the latency, local caching and local average energy consumption constraints. In particular, we first show the NP-hardness of the problem and transform it to a DC problem without loss of equivalence, which is solved efficiently via CCCP together with ADMM. In the symmetric scenario, we obtain the optimal joint policy and the closed form expressions for local caching and computing gain as well as multicast gain. In summary, we show theorectically that: in the symmetric scenario,
\begin{itemize}
  %\item increasing the computing frequency of mobile VR p$f_1$ does not always reduce the average transmission rate $R^*$;
  \item $\frac{B^*}{B_{MEC}^*}$ decreases with $C$;
    \item $\frac{B^*}{B_{MEC}^*}$ increases with $f_1$ when $\alpha > 1$ and $f_1 > \sqrt{\frac{F\bar{E}}{\mu wC}}$;
        \item $\frac{B^*}{B_{MEC}^*}$ first decreases and then increases with $f_1$ when $\alpha > 1$ and $\frac{Iw}{(1-\frac{1}{\alpha})\tau} < f_1 \leq \sqrt{\frac{F\bar{E}}{\mu wC}}$;
  \item $\frac{B^*}{B_{MEC}^*}$ remains unchanged with $f_1$ when $\alpha \leq 1$ or when $\alpha > 1$ and $f_1 \leq \frac{Iw}{(1-\frac{1}{\alpha})\tau}$;

 \item $\frac{B^*}{B_{unicast}^*}$ decreases with $K$.
\end{itemize}
\section*{Appendix A: Proof of Lemma~\ref{sub}}\label{aaaaaaaa}
Monotonicity is obvious since any caching of a new file cannot increase the value of the objective function.  In order to show the submodularity of the objective function, it is enough to prove that for each $n\in \mathcal{N}$ and $f\in \mathcal{F}$, $\frac{O_f}{\tau}\!\max_{k\in \mathcal{K}:A_{n,k}=f} \frac{1-x_{f,1}^k}{\log\left(1+\frac{Ph_k^2}{\sigma^2}\right)}$ is a submodular function since the sum of submodular functions is submodular \cite{femto}. It is direct to see that the marginal value of  $\frac{O_f}{\tau}\!\max_{k\in \mathcal{K}:A_{n,k}=f} \frac{1-x_{f,1}^k}{\log\left(1+\frac{Ph_k^2}{\sigma^2}\right)}$ for adding a new file decreases with $x_{f,1}^k$. Thus, the objective function of Problem~4 is a nonincreasing submodular function.  In addition, the constraints of Problem~4 can be rewritten as multiple matroid constraints according to \cite{femto} directly. The proof ends.
\section*{Appendix B: Proof of Lemma~\ref{symmetric}}\label{proofpolicy}
First, in the symmetric scenario,  Problem~\ref{Prob1} can be rewritten as
\begin{Prob}[Optimization Problem in Symmetric Scenario]\label{symmetricprob}
\begin{align}
&\min_{\textbf{x}}\ \ \ \frac{1}{F^K}\frac{1}{\log\left(1+\frac{Ph^2}{\sigma^2}\right)} \sum_{\textbf{A}\in \mathcal{F}^K}\sum_{f=1}^F  \Bigg[R_3 \max_{k\in \mathcal{K}} \textbf{1}(A_k=f)x_{f,3}^k\nonumber\\
&\hspace{11mm}+R_4 \max_{k\in \mathcal{K}} \textbf{1}(A_k=f)(1-\sum_{j=1}^3x_{f,j}^k)\Bigg]\nonumber\\
&\ s.t.\ \ \  \sum_{j=1}^3 x_{f,j}^k \leq 1,\ f\in \mathcal{F},\ k \in \mathcal{K},\\
&\ \ \ \ \ \ \  \sum_{f=1}^F \alpha x_{f,1}^k + x_{f,2}^k \leq \frac{C}{I},\ k\in \mathcal{K},\\
&\ \ \ \ \ \ \  \sum_{f=1}^Fx_{f,2}^k+x_{f,3}^k \leq \frac{F\bar{E}}{\mu Iwf_1^2},\ k\in \mathcal{K},\\
&\ \ \ \ \ \ \ \  x_{f,j}^k \in \{0,1\},\ f\in \mathcal{F},\ j\in \{1,2,3\},\ k\in \mathcal{K}.
\end{align}
\end{Prob}
\noindent
And $x_{f,4}^k$ can be obtained from $x_{f,4}^k = 1-\sum_{j=1}^3x_{f,j}^k$, for all $f\in \mathcal{F}$ and $k\in \mathcal{K}$.

%Secondly, we prove that the objective function of Problem~\ref{symmetricprob} is a monotonically submodular set function subject to multiple linear constraints \cite{submodular}. Based on the definition of submodular set function, it is direct to observe that $R_3 \max_{k\in \mathcal{K}} \textbf{1}(A_k=f)x_{f,3}^k$ and $R_4 \max_{k\in \mathcal{K}} \textbf{1}(A_k=f)x_{f,4}^k = R_4 \max_{k\in \mathcal{K}} \textbf{1}(A_k=f)\left(1-x_{f,1}^k-x_{f,2}^k-x_{f,3}^k\right)$ belong to submodular set functions. We omit the details due to the page limitation. Recalling the closedness property of submodular function that non-negative linear combination of submodular set functions are still submodular, the submodularity of the objective function of Problem~\ref{symmetricprob} is proved.

%based on the submodularity of the objective function of Problem~\ref{symmetricprob},

Then,  we show the optimality of symmetric policy, i.e., $x_{f,j}^{k_1} = x_{f,j}^{k_2}$, for all $f\in \mathcal{F}$, $j \in \{1,2,3,4\}$, $k_1 \in \mathcal{K}$ and $k_2 \in \mathcal{K}$. In addition, since the parameters of all the tasks are uniformly distributed in such scenario,  it is equivalent to show that
$\check{\textbf{x}} \triangleq (\check{x}_{f,j}^k)_{f\in \mathcal{F},j\in \{1,2,3,4\},k\in \mathcal{K}}$ given by
\begin{equation}\label{num3D}
\check{x}_{f,j}^k =
\begin{cases}
1& \text{$f = \sum_{j'=1}^{\max\{j-1,1\}}n_{j'}+1, \cdots, \sum_{j'=1}^{j}n_{j'}$,}\\
0& \text{otherwise,}\\
\end{cases}
\end{equation}
is without loss of optimality, where $n_j$ represents the number of tasks that are served via Route~$j$ at each mobile device.
%Specifically, denote with $n_{k,j} \triangleq \sum_{f=1}^F x_{f,j}^k$ the number of tasks that are served via Route~$j\in \{1,2,3,4\}$ at the mobile device $k$, and thus we have $\sum_{j=1}^4 n_{k,j}=F$. In particular, for mobile device $1$, given $n_{1,j}$ for all  $j\in \{1,2,3,4\}$, we assume that the corresponding $(x_{f,j}^1)_{f\in \mathcal{F}, j \in \{1,2,3,4\}}$ is given by
%\begin{equation}\label{num3D}
%x_{f,j}^1 =
%\begin{cases}
%1& \text{$f = \sum_{j'=1}^{\max\{j-1,1\}}n_{k,j'}+1, \cdots, \sum_{j'=1}^{j}n_{k,j'}$,}\\
%0& \text{otherwise,}
%\end{cases}
%j \in \{1,2,3,4\},
%\end{equation}
%which is without loss of optimality, due to the fact that the parameters of all the tasks are uniformly distributed in the symmetric scenario. Accordingly, in order to show the symmetric optimality, it is equivalent to prove that $(x_{f,j}^k)_{f\in \mathcal{F}, j\in \{1,2,3,4\}}$ for all $k\in \mathcal{K}$ is also given by (\ref{num3D}). That is, $\check{\textbf{x}} \triangleq (\check{x}_{f,j}^k)_{f\in \mathcal{F},j\in \{1,2,3,4\},k\in \mathcal{K}}$ given by
%\begin{equation}\label{num3D}
%\check{x}_{f,j}^k =
%\begin{cases}
%1& \text{$f = \sum_{j'=1}^{\max\{j-1,1\}}n_{j'}+1, \cdots, \sum_{j'=1}^{j}n_{j'}$,}\\
%0& \text{otherwise,}
%\end{cases}
%j \in \{1,2,3,4\},\ k \in \mathcal{K},
%\end{equation}
%is without loss of optimality, where $n_j$ represents the number of tasks that are served via Route~$j$ at each mobile device.
In the following, we prove (\ref{num3D}) based on contradiction. Specifically, for any $\bar{\textbf{x}} \triangleq (\bar{x}_{f,j}^k)_{f\in \mathcal{F},j\in \{1,2,3,4\},k\in \mathcal{K}}$ satisfying that: i) there exists $k' \in \mathcal{K}$, $f^1\in \mathcal{F} $ and  $f^2 \in \mathcal{F}$ such that $f^1 \leq f^2$, $\bar{x}_{f^1,j^1}^{k'} = 0$, $\bar{x}_{f^1,j^2}^{k'} = 1$, $\bar{x}_{f^2,j^1}^{k'} = 1$ and $\bar{x}_{f^2,j^2}^{k'} = 0$, where $j^i \triangleq \arg_{j\in \{1,2,3,4\}}\ \textbf{1}(\check{x}_{f^i,j}^{k'}=1)$; ii) $\bar{x}_{f,j}^k = \check{x}_{f,j}^k$ otherwise, we have
\begin{align}
&B(\bar{\textbf{x}}) - B(\check{\textbf{x}}) \nonumber\\
&= \frac{1}{F^K}\frac{1}{\log\left(1+\frac{Ph^2}{\sigma^2}\right)} \nonumber\\
&\Bigg[\sum_{\textbf{A} \in \left \{\mathcal{F}^K:A_{k'} = f^1\right\}} \Big(W(\textbf{A},f^1,\bar{\textbf {x}}) -W(\textbf{A},f^1,\check{\textbf{x}}) \Big) \nonumber\\
&+ \sum_{\textbf{A} \in \left \{\mathcal{F}^K:A_{k'} = f^2\right\}}\Big( W(\textbf{A},f^2,\bar{\textbf {x}}) -W(\textbf{A},f^2,\check{\textbf{x}})\Big)\Bigg],%\nonumber
\end{align}
where $ W(\textbf{A},f^i,\textbf{x}) \triangleq R_3 \max_{k\in \mathcal{K}} \textbf{1}(A_k=f^i) x_{f^i,3}^k + R_4  \max_{k\in \mathcal{K}} \textbf{1}(A_k=f^i) \left(1- \sum_{j=1}^3 x_{f^i,j}^k\right)$.
In the sequel, we analyze the positivity of $B(\bar{\textbf{x}}) - B(\check{\textbf{x}})$ in the following cases:
\begin{itemize}
\item If $f^1 \leq n_1, f^2 \leq n_1$, then we have $j^1=j^2=1$. Thus, $\bar{\textbf{x}} = \check{\textbf{x}}$ and $B(\bar{\textbf{x}}) - B(\check{\textbf{x}})=0$;

\item  If $f^1 \leq n_1, n_1+1 \leq f^2 \leq n_1+n_2$, then we have $j^1=1$ and $j^2 = 2$. Thus, $W(\textbf{A},f^i,\check{\textbf{x}})=W(\textbf{A},f^i,\bar{\textbf{x}})=0$, for all $i \in \{1,2\}$ and $\textbf{A} \in \left \{\mathcal{F}^K:A_{k'} = f^i\right\}$, and $B(\bar{\textbf{x}}) - B(\check{\textbf{x}})=0$;

\item If $f^1 \leq n_1, n_1+n_2+1 \leq f^2 \leq n_1+n_2+n_3$, then we have $j^1=1$ and $ j^2 =3$. Accordingly, $W(\textbf{A},f^1, \bar{\textbf{x}}) = R_3$, $W(\textbf{A},f^1,\check{\textbf{x}})=0$, $W(\textbf{A},f^2,\bar{\textbf{x}})= R_3 \max_{k\in \mathcal{K} \setminus k'} \textbf{1}\{A_k=f^2\} \bar{x}_{f^2,3}^k$ and $W(\textbf{A},f^2,\check{\textbf{x}}) = R_3$. Thus, $B(\bar{\textbf{x}})\! - B(\check{\textbf{x}})= \frac{1}{F^K}\frac{1}{\log\left(1+\frac{Ph^2}{\sigma^2}\right)} \sum_{\textbf{A} \in \left \{\mathcal{F}^K: A_{k'}=f^2 \right\} } R_3 \max_{k\in \mathcal{K} \setminus k'} \textbf{1}\{A_k\!=f^2\} \bar{x}_{f^2,3}^k \geq 0$;

\item If $f^1 \leq n_1, n_1+n_2+n_3+1 \leq f^2 \leq n_1+n_2+n_3+n_4$, then we have $j^1=1$ and $ j^2 =4$. Accordingly, $W(\textbf{A},f^1, \!\bar{\textbf{x}})\! = R_4$, $W(\textbf{A},f^1,\check{\textbf{x}})\!=0$, $W(\textbf{A},f^2,\bar{\textbf{x}})=\! R_4 \max_{k\in \mathcal{K} \setminus k'} \textbf{1}\{A_k\!=f^2\}\! \bar{x}_{f^2,4}^k$ and $W(\textbf{A},f^2,\check{\textbf{x}}) \!=\! R_4$. Thus, $B(\bar{\textbf{x}}) - B(\check{\textbf{x}})= \frac{1}{F^K}\frac{1}{\log\left(1+\frac{Ph^2}{\sigma^2}\right)} \sum_{\textbf{A} \in \left \{\mathcal{F}^K: A_{k'}=f^2 \right\} } R_4 \max_{k\in \mathcal{K} \setminus k'} \textbf{1}\{A_k\!=f^2\} \bar{x}_{f^2,4}^k \geq 0$;

\item If $n_1+1 \leq f^1 \leq n_1+n_2$, similar to the cases mentioned above, we have $B(\bar{\textbf{x}}) - B(\check{\textbf{x}}) \geq 0$;

\item If $n_1+n_2+1 \leq f^1 \leq n_1+n_2+n_3, n_1+n_2+1 \leq f^2 \leq n_1+n_2+n_3$, then we have $j^1=j^2=3$. Thus, $\bar{\textbf{x}} = \check{\textbf{x}}$ and $B(\bar{\textbf{x}}) - B(\check{\textbf{x}})=0$;

\item If $n_1+n_2+1 \leq f^1 \leq n_1+n_2+n_3, n_1+n_2+n_3+1 \leq f^2 \leq n_1+n_2+n_3+n_4$, then we have $j^1=3$ and $j^2=4$. Accordingly, $W(\textbf{A},f^1, \bar{\textbf{x}}) = R_3 \max_{k\in \mathcal{K} \setminus k'} \textbf{1}\{A_k=f^1\} \bar{x}_{f^1,3}^k + R_4$, $W(\textbf{A},f^1,\check{\textbf{x}})=R_3 $, $W(\textbf{A},f^2,\bar{\textbf{x}})= R_3+R_4 \max_{k\in \mathcal{K} \setminus k'} \textbf{1}\{A_k=f^2\} \bar{x}_{f^2,4}^k$ and $W(\textbf{A},f^2,\check{\textbf{x}}) = R_4$. Thus, $B(\bar{\textbf{x}}) - B(\check{\textbf{x}})= \frac{1}{F^K}\frac{1}{\log\left(1+\frac{Ph^2}{\sigma^2}\right)} \Bigg[\sum_{\textbf{A} \in \left \{\mathcal{F}^K:A_{k'} = f^1\right\}} R_3 \max_{k\in \mathcal{K} \setminus k'} \textbf{1}\{A_k=f^1\} \bar{x}_{f^1,3}^k + \sum_{\textbf{A} \in \left \{\mathcal{F}^K:A_{k'} = f^2\right\}} R_4 \max_{k\in \mathcal{K} \setminus k'} \textbf{1}\{A_k=f^2\} \bar{x}_{f^2,4}^k\Bigg] \geq 0$.
\end{itemize}
Thus, we have $B(\bar{\textbf{x}}) - B(\check{\textbf{x}}) \geq 0$, which contradicts the optimality of $\check{\textbf{x}}$, and the optimality of the symmetric assumption in (\ref{num3D}) holds.

Next, based on the symmetric property of the joint policy in (\ref{num3D}), the objective function of Problem~\ref{symmetricprob}, i.e., $B(\textbf{x}) $, can be rewritten as:
\begin{align}\label{equation}
&B(\textbf{x})\nonumber\\ &\overset{(a)}{=}\frac{1}{\log\left(1\!+\frac{Ph^2}{\sigma^2}\right)}\! \sum_{\textbf{A}\in \mathcal{F}^K} P(\textbf{A}) \!\Bigg(\sum_{f=n_1+n_2+1}^{n_1+n_2+n_3} R_3\nonumber\\&\hspace{5mm} \max_{k\in \mathcal{K}}\textbf{1}(A_k=f) + \!\sum_{f=n_1+n_2+n_3+1}^{F}\!R_4\! \max_{k\in \mathcal{K}} \textbf{1}(A_k=f)\Bigg)\nonumber\\
&\overset{(b)}{=} \frac{1}{\log\left(1+\frac{Ph^2}{\sigma^2}\right)} \Bigg[R_3\sum_{f=n_1+n_2+1}^{n_1+n_2+n_3} \sum_{(\pi_{f,k})_{k\in \mathcal{K}} \in \{0,1\}^K} \nonumber\\
&\hspace{35mm}P\left((\pi_{f,k})_{k\in \mathcal{K}}\right) \textbf{1}\left(\sum_{k\in \mathcal{K}} \pi_{f,k}>1\right)\nonumber\\
& + R_4\sum_{f=n_1+n_2+n_3+1}^{F} \sum_{(\pi_{f,k})_{k\in \mathcal{K}} \in \{0,1\}^K} P\left((\pi_{f,k})_{k\in \mathcal{K}}\right)\nonumber\\
& \hspace{40mm} \textbf{1}\left(\sum_{k\in \mathcal{K}} \pi_{f,k}>1\right)\Bigg]\nonumber\\
%& \overset{(c)}{=} R_3n_3 \left(1-\Pr\{()\}\right)\nonumber\\
&= \frac{1}{\log\left(1+\frac{Ph^2}{\sigma^2}\right)} \left( 1\!-(1\!-\frac{1}{F})^K\right)\nonumber\\
&\hspace{25mm} \left(R_3n_3\! +\! R_4 \left(F-n_1-n_2-n_3\right)\right),
\end{align}
where $P\left((\pi_{f,k})_{k\in \mathcal{K}}\right) \triangleq \frac{1}{F}\pi_{f,k} + (1-\frac{1}{F})(1- \pi_{f,k})$. Specifically, (a) is directly obtained from (\ref{num3D}), and  (b) has been proved in \cite{hete}. Accordingly, Problem~\ref{symmetricprob} can be rewritten as
\begin{Prob}[Optimization Problem in Symmetric Scenario]\label{symmetricprob2}
\begin{align}
&\min_{n_1,n_2,n_3}\ \ \ \frac{1}{\log\left(1+\frac{Ph^2}{\sigma^2}\right)}\left( 1-(1-\frac{1}{F})^K\right) \Big(R_3n_3\nonumber\\
&\hspace{40mm} + R_4 \left(F-n_1-n_2-n_3\right)\Big)\nonumber\\
%&\ s.t. \hspace{5cm}\sum_{j=1}^3 x_{f,j}^k \leq 1,\ f\in \mathcal{F},\ k \in \mathcal{K},\\
&\ s.t. \hspace{10mm}   \alpha n_1+ n_2 \leq \frac{C}{I},\\
&\ \ \ \ \hspace{10mm}n_2+n_3 \leq \frac{F\bar{E}}{\mu Iwf_1^2},\\
&\ \ \ \ \hspace{10mm} n_1+n_2+n_3 \leq F,\\
&\ \ \ \ \hspace{10mm}n_1 \geq 0,\ n_2 \geq 0,\ n_3 \geq 0.
%&\ \ \ \ \hspace{5cm} x_{f,j}^k \in \{0,1\},\ f\in \mathcal{F},\ j\in \{1,2,3\},\ k\in \mathcal{K}.
\end{align}
\end{Prob}
Note that Problem~\ref{symmetricprob2} is a linear programming, and the solution can be trivially obtained.  The proof ends.
\section*{Appendix C: Proof of Theorem~\ref{MECgain}}\label{prooftheo1}
$B_{MEC}^*$ and $B^*$ can be obtained directly from (\ref{equation}). Specifically, for MEC computing, we have  $n_1=n_2=n_3=0$ and $n_4=F$, and thus
\begin{equation}
B_{MEC}^* = \frac{1}{\log\left(1+\frac{Ph^2}{\sigma^2}\right)} \left( 1\!-(1\!-\frac{1}{F})^K\right) FR_4.
\end{equation}
For $B^*$, similarly, from Lemma~\ref{symmetric}, we can obtain the optimal value of $n_i$ for all $i\in \{1,2,3,4\}$. By substituting $n_i$ for all $i\in \{1,2,3,4\}$ into (\ref{equation}), we can directly obtain $B^*$. The proof ends.

\section*{Appendix D: Proof of Theorem~\ref{gain1}}\label{prooftheo1}
In the symmetric scenario, $B_{unicast}^*$ can be obtained directly from \cite{Sunvr}, and $B^*$ can be obtained as mentioned above. The proof ends.

\end{document}